\documentclass[draftcls,journal,onecolumn]{IEEEtran}

\usepackage[utf8]{inputenc} 
\usepackage[T1]{fontenc}    
\usepackage{hyperref}       
\usepackage{url}            
\usepackage{booktabs}       
\usepackage{amsfonts}       
\usepackage{nicefrac}       
\usepackage{microtype}      
\usepackage{amsmath}
\usepackage{amssymb}
\usepackage{amsthm}
\usepackage{bm}
\usepackage{bbm}
\usepackage{graphicx}
\usepackage{comment}
\usepackage[dvipsnames]{xcolor}
\usepackage{algorithm}
\usepackage{algorithmic}
\usepackage{support-caption}
\usepackage{subcaption}
\usepackage{mathrsfs}
\usepackage{mathtools,xparse}
\definecolor{mygreen}{rgb}{0.75,0.8,0.72}
\newcommand{\RED}[1]{{\color{black}#1}}
\makeatletter
\newcommand{\printfnsymbol}[1]{%
  \textsuperscript{\@fnsymbol{#1}}%
}
\makeatother
\newcommand{\diag}{\mbox{diag}}
\DeclarePairedDelimiter\abs{\lvert}{\rvert}

\DeclarePairedDelimiter{\bra}{\lbrack}{\rbrack}
\DeclarePairedDelimiter{\bre}{\lbrace}{\rbrace}
\DeclarePairedDelimiter{\para}{(}{)}
\DeclarePairedDelimiter{\dotp}{\langle}{\rangle}

\newtheorem{theorem}{Theorem}
\newtheorem{lemma}{Lemma}

\newtheorem{remark}{Remark}

\DeclareMathOperator{\E}{\mathsf{E}}
\DeclareMathOperator{\pr}{\mathsf{P}}
\title{Robust Mean Estimation in High Dimensions: An Outlier-Fraction Agnostic and Efficient Algorithm}

\author{\resizebox{1\hsize}{!}{Aditya Deshmukh\printfnsymbol{1},
\textit{Student Member, IEEE}, Jing Liu\printfnsymbol{1},
\textit{Member, IEEE},
Venugopal V. Veeravalli,
\textit{Fellow, IEEE}
\thanks{\printfnsymbol{1}Equal contribution.}
\thanks{A. Deshmukh and V.V. Veeravalli are with the ECE Department and Coordinated Science Lab, University of Illinois at Urbana-Champaign, Illinois,
USA. Email: ad11,vvv@illinois.edu}
\thanks{J. Liu is with the Mitsubishi Electric Research Laboratories, Massachusetts,
USA. Email: jiliu@merl.com}
\thanks{This research was supported by the US Army Research Laboratory under Cooperative Agreement W911NF-17-2-0196 and by the US National Science Foundation under grant 2106727, through the University of Illinois at Urbana-Champaign.}
}}
\IEEEoverridecommandlockouts

\begin{document}
\maketitle

\begin{abstract}
 The problem of robust mean estimation in high dimensions is studied, in which a certain fraction (less than half) of the datapoints can be arbitrarily corrupted. Motivated by compressive sensing, the robust mean estimation problem is formulated as the minimization of the  $\ell_0$-`norm' of an \emph{outlier indicator vector}, under a second moment constraint on the datapoints. The $\ell_0$-`norm' is then relaxed to the $\ell_p$-norm ($0<p\leq 1$) in the objective, and it is shown that the global minima for each of these objectives are order-optimal and have optimal breakdown point for the robust mean estimation problem. Furthermore, a computationally tractable iterative $\ell_p$-minimization and hard thresholding algorithm is proposed that outputs an order-optimal robust estimate of the population mean. The proposed algorithm (with breakdown point $\approx 0.3$) does not require prior knowledge of the fraction of outliers, in contrast with most existing algorithms, and for $p=1$ it has near-linear time complexity. Both synthetic and real data experiments demonstrate that the proposed algorithm outperforms state-of-the-art robust mean estimation methods. 
 
\end{abstract}

\section{Introduction}

Robust mean estimation in high dimensions has received considerable interest recently, and has found applications in areas such as data analysis (e.g., spectral data in astronomy~\cite{10.2307/1271538}), outlier detection~\cite{huber2011robust,maronna2018robust,NIPS2019_8839} and distributed machine learning~\cite{10.1145/3154503,yin2018byzantine,bubeck2013bandits}. Classical robust mean estimation methods such as coordinate-wise median and geometric median have error bounds that scale with the dimension of the data \cite{LRV}, which results in poor performance in the high dimensional regime. A notable exception is Tukey's Median~\cite{Tukey1975MathematicsAT} that has an error bound that is independent of the dimension, when the fraction of outliers is less than a threshold 
~\cite{donoho1992breakdown,zhu2020does}. However, the computational complexity of Tukey's Median algorithm is exponential in the dimension. 

A number of recent papers have proposed polynomial-time algorithms that have dimension independent error bounds under certain distributional assumptions (e.g., bounded covariance or concentration properties). For a recent comprehensive survey on robust mean estimation, we refer the interested readers to~\cite{diakonikolas2019recent}.
One of the first such algorithms is Iterative Filtering~\cite{7782980,diakonikolas2017being,steinhardt2018robust}, in which one finds the top eigenvector of the sample covariance matrix and removes (or down-weights) the points with large projection scores on that eigenvector, and then repeat this procedure on the rest of points until the top eigenvalue is small. However, as discussed in~\cite{NIPS2019_8839}, the drawback of this approach is that it only looks at one direction/eigenvector at a time, and the outliers may not exhibit unusual bias in only one direction or lie in a single cluster. Figure~\ref{fig:LpApproximation} illustrates an example for which Iterative Filtering might have poor empirical performance. In this figure, the inlier datapoints in blue are randomly generated from the standard Gaussian distribution in (high) dimension $d$, and therefore their $\ell_2$-distances to the origin are roughly $\sqrt{d}$ (see, e.g., Theorem 3.1 of~\cite{adams2020ma3k0}). There are two clusters of outliers in red, and their $\ell_2$-distances to the origin are also roughly $\sqrt{d}$. If there is only one cluster of outliers, Iterative Filtering can effectively identify them; however, in this example, this method may remove many inlier points and perform suboptimally.
\begin{figure*}[!ht]
	\centering
	\includegraphics[trim=280 0 0 0, clip=true,width=0.24\linewidth]{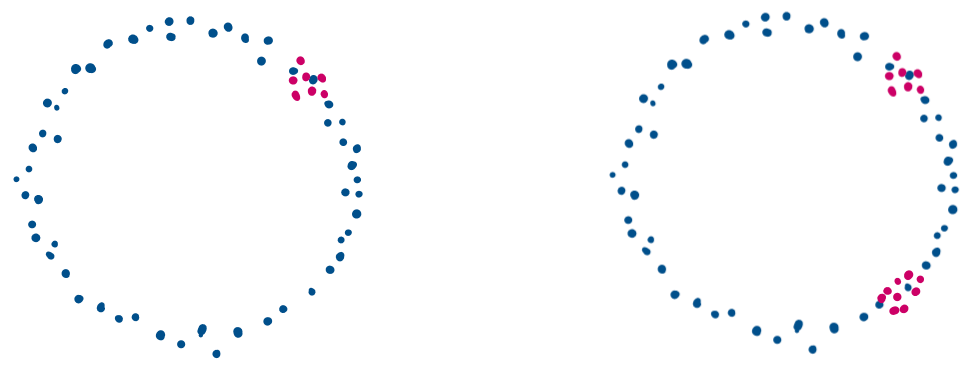}
	\captionsetup{justification=centering}
	\caption{Illustration of two clusters of outliers (red points). The inlier points (blue) are drawn from standard Gaussian distribution in high dimension $d$. Both the outliers and inliers are at roughly $\sqrt{d}$ distance from the origin.}
	\label{fig:LpApproximation}
\end{figure*}

There are interesting connections between existing methods for robust mean estimation and those used in compressive sensing. The Iterative Filtering algorithm has similarities to the greedy Matching Pursuit compressive sensing  algorithm~\cite{Mallat1993MatchingPW}. In the latter algorithm, one finds a single column of sensing matrix $\bm A$ that has largest correlation with the measurements $\bm b$, removes that column and its contribution from $\bm b$, and repeats this procedure on the remaining columns of $\bm A$. 
Dong et al.~\cite{NIPS2019_8839} proposed a new scoring criteria for finding outliers, in which one looks at multiple directions associated with large eigenvalues of the sample covariance matrix in every iteration of the algorithm. Interestingly, this multi-directional approach is conceptually similar to  Iterative Thresholding techniques in compressive sensing (e.g., Iterative Hard Thresholding~\cite{blumensath2009iterative} or Hard Thresholding Pursuit~\cite{foucart2011hard}), in which one simultaneously finds multiple columns of matrix $\bm A$ that are more likely contribute to $\bm b$. Although iterative thresholding techniques are also greedy, they are more accurate than the Matching Pursuit technique in practice~\cite{blumensath2008iterative,BOUCHOT2016412}.

A common assumption in robust mean estimation problem is that the fraction of the corrupted datapoints is small. In this paper, we explicitly use this information 
through the introduction of an  \textit{outlier indicator vector} whose  $\ell_0$-`norm' we minimize under a second moment constraint on the datapoints. This is partially motivated by compressive sensing and shares the same principle of `fitting the majority of the data' that is common in robust statistics. This new formulation not only enables us to leverage  advanced compressive sensing techniques to solve the robust mean estimation problem, but also allow us to design algorithms that do not require prior knowledge of the fraction of outliers.  {There are some works in sparse recovery (see, e.g. \cite{wang2010, candes2005}), in which $\ell_0/\ell_p$ minimization is used to remove outliers in data. In these works, a linear model $y=Ax+e$ is considered, wherein $y$ denotes the measurements, the matrix $A$ is known, and the unknown sparse vector $e$ models the potential outlier corruption on each datapoint.  Consequently, the analyses in the works on sparse recovery methods heavily rely on the assumption that the underlying model is linear (e.g., some works exploit the range-space/null-space properties of the matrix $A$). On the other hand, in robust mean estimation, a general observation model (not necessarily linear) is considered. In light of this, the analyses in the works on sparse recovery cannot be transferred in an obvious way to the robust mean estimation problem.}


We consider the setting in which the distribution of the datapoints before corruption has bounded covariance, as is commonly assumed in many recent works (e.g.,~\cite{diakonikolas2017being,NIPS2019_8839,convex,steinhardt2017resilience}). 
In particular, in ~\cite{convex}, the authors propose to minimize the spectral norm of the weighted sample covariance matrix and use the knowledge of the outlier-fraction $\epsilon$ to constrain the weights. Along this line, two very recent works~\cite{cheng2020high,zhu2020robust} show that any approximate stationary point of the objective in~\cite{convex} gives a near-optimal solution. In contrast, our objective is designed to minimize the sparsity of an \emph{outlier indicator vector}, and we show that \textit{any} sparse enough solution is nearly optimal.

\noindent \emph{Contributions:}
\begin{itemize}
\item At a fundamental level, a contribution of this paper is the formulation of the robust mean estimation problem as minimizing the $\ell_0$-`norm' of the proposed \emph{outlier indicator vector}, under a second moment constraint on the datapoints. In addition, order-optimal estimation error guarantees and optimal breakdown point $(\epsilon<1/2)$ are shown for this objective. We relax the $\ell_0$ objective to $\ell_p (0< p\leq 1)$ as in compressive sensing, and establish corresponding order-optimal estimation error guarantees. {The guarantees are order-optimal with respect to the number of datapoints($n$), dimension of the data ($d$), and the fraction of corrupted datapoints($\epsilon$). Henceforth we use the term `order-optimal' in this sense.}
%

\item Motivated by the proposed $\ell_0$ and $\ell_p$ objectives and their theoretical justifications, we propose a computationally tractable \emph{iterative} $\ell_p (0<p\leq 1)$ minimization and hard thresholding algorithm, and establish the order optimality of the algorithm. Empirical studies show that the proposed algorithm significantly outperforms state-of-the-art methods in robust mean estimation.

\item The proposed algorithm (with maximal breakdown point of $1-1/\sqrt{2}$) does not require the knowledge of the fraction of outliers (in contrast to most existing algorithms). For $p=1$, the algorithm has near-linear time complexity.
\end{itemize}

\section{Proposed optimization problems}\label{sec:propose}
We begin by defining what we mean by a corrupted sample of datapoints.
\newtheorem{definition}{Definition}
\begin{definition}\label{def1}
($\epsilon$-corrupted sample~\cite{NIPS2019_8839}) Let $P$ be a distribution on $\mathbb{R}^d$ with unknown mean $\bm\mu$, and let $\tilde{\bm y}_1,...,\tilde{\bm y}_n$ be independent and identically distributed (i.i.d.) drawn from $P$. These datapoints are then modified by an adversary who can inspect all the datapoints, remove $\epsilon n$ of them, and replace them with arbitrary vectors in $\mathbb{R}^d$. We then obtain an $\epsilon$-corrupted sample, denoted as $\bm y_1,...,\bm y_n$.

\end{definition}
Throughout the rest of the paper, we adhere to the notation given above: we represent a datapoint before corruption as $\tilde{\bm y}_i$, and after corruption as $\bm y_i$. Given a set of datapoints $\{\bm x_i, i=1,\dots,n\}$, we term the following as \textit{sample covariance matrix around $\bm z$}:
\begin{align}
    \sum\limits_{i=1}^n (\bm x_i-\bm z)(\bm x_i-\bm z)^\top.
\end{align}

There are other types of contamination one can consider, e.g., Huber's $\epsilon$-contamination model~\cite{huber1964}. The contamination model described in Definition \ref{def1} is the strongest in the sense that the adversary is not oblivious to the original datapoints, and can replace any subset of $\epsilon n$ datapoints with any vectors in $\mathbb{R}^d$. We refer the reader to~\cite{diakonikolas2019recent} for a more detailed discussion on contamination models.

Our primary goal is to robustly estimate the true population mean, given an $\epsilon$-corrupted sample. We assume that the underlying distribution has bounded second moment. A powerful and useful key insight that was exploited in previous work on the problem is that if the outliers in an $\epsilon$-corrupted sample (of large size) shift the average of datapoints before corruption by $\Omega(\xi)$ in a direction $\nu$, then the variance of the projected sample along $\nu$ increases by $\Omega(\xi^2/\epsilon)$. Thus, intuitively, it suffices to find a large subset of the $\epsilon$-corrupted sample, whose sample covariance matrix is close to the covariance matrix of the underlying distribution. In order for such a subset to exist and for the mean of this large subset to be close to the true mean, we need some form of concentration of the datapoints (before corruption) around the mean of their distribution. A constrained second moment condition is sufficient to guarantee this, and such an assumption is also used in previous works. In the following, we provide a brief high-level explanation (details can be found in the Appendix). Suppose we are given a sufficiently large sample of datapoints of size $n$, generated from a distribution with mean $\bm\mu$ and spectral norm of the covariance matrix bounded by $\sigma^2$.  Then, with high probability, there exists a large subset of the sample with spectral norm of the sample covariance matrix around $\bm\mu$ bounded by $O(\sigma^2)$. Hence, after corruption, with high probability there still exists a sufficiently large subset, say $\bm G^*$, of the resulting $\epsilon$-corrupted sample, of size $(1-\epsilon')n$ (where $\epsilon'\to\epsilon$ as $n\to\infty$), such that the spectral norm of the sample covariance matrix around $\bm\mu$ is bounded by $O(\sigma^2)$. Utilizing this, the concentration of the sample before corruption around $\bm\mu$, and a fundamental result~\cite[Lemma C.2]{zhu2020robust} about closeness of population mean and conditional mean, it can be shown that the distance between $\bm\mu$ and the sample average of $\bm G^*$ is $O(\sigma\sqrt{\epsilon'})$.


Based on this motivation, we propose an $\ell_0$-minimization problem to find the largest subset, whose sample covariance matrix exhibits bounded spectral norm. We first introduce an \textit{outlier indicator vector} $\bm h$: for the $i$-th datapoint, $h_i$ indicates that whether it is an outlier ($h_i = 1$) or not ($h_i= 0$). Given an $\epsilon$-corrupted sample of size $n$, we propose the following optimization problem, for which the solution in $\bm x$ should yield a robust estimate of the mean:
\begin{align}\label{obj_01}
   \min_{\bm h, \bm x} \|\bm h\|_0 \quad
    s.t.\  &   h_i \in\{0, 1\}, \forall i,\\
   & \lambda_{\mathrm{max}}\left(\sum_{i=1}^{n} (1-h_i)(\bm y_i-\bm x)(\bm y_i-\bm x)^\top \right)\leq c_1^2 \sigma^2n, \nonumber
\end{align}
{where $c_1$ is a constant that controls the inflation of the constraint with respect to the bound $(\sigma^2)$ on the spectral norm of the covariance matrix of the underlying distribution.} 

 We further relax the problem to the following:
\begin{align}\label{obj}
   \min_{\bm h, \bm x} \|\bm h\|_0 \quad
    s.t.\  &   0\leq h_i \leq 1, \forall i,\\
   & \lambda_{\mathrm{max}}\left(\sum_{i=1}^{n} (1-h_i)(\bm y_i-\bm x)(\bm y_i-\bm x)^\top \right)\leq c_1^2 \sigma^2n. \nonumber
\end{align}
 Note that any globally optimal solution of \eqref{obj_01} is also globally optimal solution of \eqref{obj}. {To see this, let $\tilde{\bm h}$ be a global optimum of \eqref{obj}. Let $\bm h'$ be the vector obtained after setting the non-zero values of $\tilde{\bm h}$ to 1. Note that $\bm h'$ has the same $\ell_0$-norm as $\tilde{\bm h}$, and is also a feasible point of \eqref{obj_01}. Since the constraint set of \eqref{obj} is larger than \eqref{obj_01}, the optimum value of \eqref{obj_01} must be greater than or equal to the optimum value of \eqref{obj_01}. This implies that $\bm h'$ is a global optimum of \eqref{obj_01}. Hence, the claim holds.}  We show in Theorem \ref{thm:l0}, that any sparse enough feasible pair including the global optimum of~\eqref{obj} achieves order-optimality in terms of the error in estimating the mean.
 
 However, minimizing the above $\ell_0$ objective is not computationally tractable. Motivated by compressive sensing, we further propose to relax the $\ell_0$-`norm' to the $\ell_p$-norm ($0<p\leq 1$), which leads to the following optimization problem:

\begin{align}\label{objL1}
   \min_{\bm h, \bm x} \|\bm h\|_p \quad
    s.t.\  &  0\leq h_i \leq 1, \forall i,\\
   & \lambda_{\mathrm{max}}\left(\sum_{i=1}^{n} (1-h_i)(\bm y_i-\bm x)(\bm y_i-\bm x)^\top \right)\leq c_1^2 \sigma^2n. \nonumber
\end{align}

We show in Theorem \ref{thm:lp2}, that even in this case any `good' feasible pair including the global optimum is order-optimal in terms of the error in estimating the mean.

{
In the approaches taken in prior works (see, e.g., \cite{zhu2020robust}), the robust mean estimation problem is the following feasibility problem:
\begin{align}\label{eq:opt_prob_1}
    \text{Find } h \text{ s.t. } h_i\in [0,1],\sum\limits_{i=1}^n h_i\leq\epsilon n, \lambda_{\max}(\sum\limits_{i=1}^n(1-h_i)(y_i-x)(y_i-x)^\top)\leq C\sigma^2n.
\end{align}
 Most works (see, e.g., \cite{NIPS2019_8839,convex}) consider the following problem (or its variant), which is obtained by changing the feasibility problem into the following optimization problem:
\begin{align}\label{eq:opt_prob_2}
    \min \lambda_{\max} (\sum\limits_{i=1}^n(1-h_i)(y_i-x)(y_i-x)^\top)\text{ s.t. }h_i\in [0,1], \sum\limits_{i=1}^n h_i\leq \epsilon n
\end{align}
where $x$ is either fixed or is the weighted average of $y_i$'s with weights as $1-h_i$. Landscape results related to the optimization problem \eqref{eq:opt_prob_2} were obtained in \cite{convex} and \cite{zhu2020robust}. Our formulation \eqref{objL1}, for the special case of $p=1$, corresponds to minimizing the feasibility condition related to the sum of "weights" in \eqref{eq:opt_prob_1}. We provide landscape results for the optimization problem given in \eqref{objL1} (Theorems \ref{thm:l0} and \ref{thm:lp2}). An advantage of our formulation, which we will exploit in Algorithm \ref{alg:1}, is that it does not require knowledge of the fraction of outliers $\epsilon$.
}

We now provide theoretical guarantees for the estimator which is given by the solution of the optimization problem \eqref{obj}. We show that given an $\epsilon$-corrupted sample of sufficiently large size, then with high probability, the $\ell_2$-norm of the estimator's error is $O\left(\sigma\sqrt{\frac{\epsilon+\frac{d\log d}{n}}{1-2\para*{\epsilon+\frac{d\log d}{n}}}}\right)$. We formalize this in the following theorem. It is well known that an information-theoretic lower bound on the $\ell_2$-norm of any estimator's error $\|\hat{\bm x}-\bm\mu\|_2$ is $\Omega\left(\sigma\sqrt{\frac{\epsilon}{1-2\epsilon}}\right)$ (see~\cite{zhu2020robust}). Thus, the estimator is order-optimal in terms of the error as $\alpha \to 0$ and $n\to\infty$.

\begin{theorem}\label{thm:l0}
Let $P$ be a distribution on $\mathbb{R}^d$ with unknown mean $\bm \mu$ and unknown covariance matrix 	$\Sigma\preceq \ \sigma^2 I$. Let $\delta\in(0,1/4)$ and $c_1>1$ be fixed. Let $c_1'=c^2_1 \min\bre*{c^2_1\log c^2_1+1-c^2_1,1}$, {$n>\frac{2e}{c_1'\delta^2}d\log(d/\delta)$ and $\alpha=\frac{ed\log(d/\delta)}{n\delta^2c'_1}$}. Let $\epsilon\in (0,1/2-\alpha)$ and $\epsilon'=\epsilon+\alpha$. Given an $\epsilon$-fraction corrupted set of $n$ datapoints from $P$, let
\begin{equation}
    \mathcal{S}=\bre*{(\bm h,\bm x):  \|\bm h\|_0< (1-\epsilon')n; {\bm x}=\frac{\sum_{\{i:{h}_i=0\}} \bm y_i}{|\{i:{h}_i=0\}|}}.
\end{equation}

Then the following holds with probability at least $1-4\delta$:\\
1) Any feasible pair $(\hat{\bm h},\hat{\bm x})$ for the optimization problem \eqref{obj} such that $(\hat{\bm h},\hat{\bm x})\in\mathcal{S}$ satisfies
    \begin{align}
    \left \| \hat{\bm x}-\bm\mu\right \|_2 \leq
        \para*{\sqrt{\frac{c^2_1\sigma^2}{1-\epsilon'}}+\sqrt{\frac{c^2_1\sigma^2}{1-\frac{\|\hat{\bm h}\|_0}{n}}}}\sqrt{\frac{\max\bre*{\epsilon',\frac{\|\hat{\bm h}\|_0}{n}}}{1-\epsilon'-\frac{\|\hat{\bm h}\|_0}{n}}}+\sqrt{\frac{c_1^2\sigma^2}{1-\alpha}.\frac{\epsilon}{1-\epsilon}}+{\sigma\sqrt{\alpha\delta}\para*{1+2\sqrt{\frac{c_1'}{e\log(\frac{d}{\delta})}}}}.\label{eq:L0Objgau}
\end{align}
\\
2) A global optimum $(\bm h^{\mathrm{opt}}, \bm x^{\mathrm{opt}})$ of \eqref{obj} lies in $\mathcal{S}$ with $\|\bm h^{\mathrm{opt}}\|_0\leq \epsilon' n$.
\end{theorem}


The proof is deferred to the Appendix. A high-level sketch of the proof of Theorems~\ref{thm:l0} is as follows. We use the idea in \cite[Lemma 2.2]{zhu2020robust} stated in Lemma \ref{lem:tv}. Informally, if two probability distributions on a set of datapoints are close in total variation distance, then the weighted means of the distributions are close. Consider the uniform distribution on the set $\{\bm y_i:\hat{h}_i=0\}$ (say $P_1$). Note that the estimator $\hat{\bm x}$ in Theorem~\ref{thm:l0} is the mean of $P_1$. We show that the total variation distance  between $P_1$ and the uniform distribution (say $P'$) on the set of inlier datapoints (that are within a distance of {$\sigma\sqrt{\frac{d}{\alpha\delta}}=\sigma\sqrt{\frac{n\delta c'_1}{e\log(d/\delta)}}$}  from $\bm\mu$), is small. Therefore one can show that the distance between $\hat{\bm x}$ and the mean of $P'$ is $O\left(\sigma\sqrt{\frac{\max\bre*{\epsilon',\frac{\|\hat{\bm h}\|_0}{n}}}{1-\epsilon'-\frac{\|\hat{\bm h}\|_0}{n}}}\right)$. Using Lemma \ref{lem:tv}, we show that the distance between the mean of $P'$ and $\bm\mu$ is $O(\sigma\sqrt{\epsilon'})$. Using triangle inequality, it follows that the distance between $\hat{\bm x}$ and $\bm\mu$ is $O\left(\sigma\sqrt{\frac{\max\bre*{\epsilon',\frac{\|\hat{\bm h}\|_0}{n}}}{1-\epsilon'-\frac{\|\hat{\bm h}\|_0}{n}}}\right)$.

\begin{remark}\label{rem:2}
Theorem~\ref{thm:l0} shows that, as long as we find a feasible point $\hat{\bm h}$ that is sparse enough, i.e., $\|\hat{\bm h}\|_0 \leq (\epsilon+\alpha)n$, the average of the estimated inliers $\frac{\sum_{\{i:\hat{h}_i=0\}} \bm y_i}{|\{i:\hat{h}_i=0\}|}$ is close to the true mean in the optimal sense. It is not necessary to reach the global optimum of the objective~\eqref{obj}.
\end{remark}

We now provide a similar order-optimal error guarantee for the solution of the optimization problem in \eqref{objL1}.

\begin{theorem}\label{thm:lp2}
Let $P$ be a distribution on $\mathbb{R}^d$ with unknown mean $\bm\mu$ and unknown covariance matrix 	$\Sigma\preceq \sigma^2 I$.Let $\delta\in(0,1/4)$, $c_1>1$ and $p\in(0,1]$ be fixed. Let $c_1'=c^2_1 \min\bre*{c^2_1\log c^2_1+1-c^2_1,1}$, {$n>\frac{2e}{c_1'\delta^2}d\log(d/\delta)$ and $\alpha=\frac{ed\log(d/\delta)}{n\delta^2c'_1}$}. Let $\epsilon\in (0,1/2-\alpha)$ and $\epsilon'=\epsilon+\alpha$. Given an $\epsilon$-fraction corrupted set of $n$ datapoints from $P$, let
\begin{align}
\mathcal{S}'=\bre*{(\bm h,\bm x):  \|\bm h\|^p_p< (1-\epsilon')n;\;\; {\bm x}=\frac{\sum_{i=1}^n (1-h_i)\bm y_i}{\sum_{i=1}^n (1-h_i)}}\label{eq:H'}.
\end{align}

Then the following holds with probability at least $1-4\delta$:
\begin{enumerate}
    \item Any feasible pair $(\hat{\bm h},\hat{\bm x})$ of \eqref{objL1} such that $(\hat{\bm h},\hat{\bm x})\in\mathcal{S}'$ satisfies
    \begin{align}
    \begin{split}
    \left \| \hat{\bm x}-\bm\mu\right \|_2 \leq &
         \para*{\sqrt{\frac{c^2_1\sigma^2}{1-\epsilon'}}+\sqrt{\frac{c^2_1\sigma^2}{1-\frac{\|\hat{\bm h}\|^p_p}{n}}}}\sqrt{\frac{\max\bre*{\epsilon',\frac{\|\hat{\bm h}\|^p_p}{n}}}{1-\epsilon'-\frac{\|\hat{\bm h}\|^p_p}{n}}}+\sqrt{\frac{c_1^2\sigma^2}{1-\alpha}.\frac{\epsilon}{1-\epsilon}}\\
         &+\sigma\sqrt{\alpha\delta}\para*{1+2\sqrt{\frac{c_1'}{e\log(\frac{d}{\delta})}}}.
         \end{split}
\end{align}
\item A global optimum $(\bm h^{\mathrm{opt}}, \bm x^{\mathrm{opt}})$ of \eqref{objL1} lies in $\mathcal{S}'$ with $\|\bm h^{\mathrm{opt}}\|^p_p\leq \epsilon' n$.
\end{enumerate}
\end{theorem}


The proof is deferred to the Appendix. The high-level idea is similar to that of the proof of Theorem \ref{thm:l0}. We consider the distribution on the $\alpha$-corrupted samples with (normalized) probability weights $1-h_i$ (say $P_2$). Note that the estimator $\hat{\bm x}$ in Theorem~\ref{thm:lp2} is the mean of $P_2$. We show that the total variation distance  between $P_2$ and the uniform distribution (say $P'$) on the set of inlier datapoints (that are within a distance of {$\sigma\sqrt{\frac{d}{\alpha\delta}}=\sigma\sqrt{\frac{n\delta c'_1}{e\log(d/\delta)}}$} from $\bm\mu$), is small. Therefore one can show that the distance between $\hat{\bm x}$ and the mean of $P'$ is $O\left(\sigma\sqrt{\frac{\max\bre*{\epsilon',\frac{\|\hat{\bm h}\|^p_p}{n}}}{1-\epsilon'-\frac{\|\hat{\bm h}\|^p_p}{n}}}\right)$. Using Lemma \ref{lem:tv}, we show that the distance between the mean of $P'$ and $\bm\mu$ is $O(\sigma\sqrt{\epsilon'})$. Using triangle inequality, it follows that the distance between $\hat{\bm x}$ and $\bm\mu$ is $O\left(\sigma\sqrt{\frac{\max\bre*{\epsilon',\frac{\|\hat{\bm h}\|^p_p}{n}}}{1-\epsilon'-\frac{\|\hat{\bm h}\|^p_p}{n}}}\right)$.

\begin{remark}
The breakdown point of the estimators in Theorems \ref{thm:l0} and \ref{thm:lp2} is nearly the maximal possible $1/2$ (as $\alpha\to 0$ and $n\to\infty$), that is the estimator can tolerate any corruption level $\epsilon<1/2$, assuming that the number of samples $n$ satisfies the lower bound.
\end{remark}

\begin{remark}\label{rem_thm2}
From Lemma \ref{lem:wmin} in the Appendix, we know that given any feasible pair  of \eqref{objL1} with $\|\hat{\bm h}\|_p\leq (\epsilon'n)^{1/p}$, we have that $\left(\hat{\bm h},\frac{\sum\limits_{i=1}^n (1-\hat{h}_i) \bm y_i}{\sum\limits_{i=1}^n (1-\hat{h}_i)}\right)$ is also a feasible pair, and therefore it lies in the set $\mathcal{S}'$ defined in~\eqref{eq:H'}. Theorem~\ref{thm:lp2} further shows that this weighted average of the datapoints $\frac{\sum\limits_{i=1}^n (1-\hat{h}_i) \bm y_i}{\sum\limits_{i=1}^n (1-\hat{h}_i)}$ is close to the true mean. Again, we note that it is not necessary to reach the global optimum of the objective~\eqref{objL1}; we only need to find a feasible point $\bm h$ of~\eqref{objL1} whose $\ell_p$-norm is small enough. %
 \end{remark}
\section{Algorithm}


\subsection{\texorpdfstring{$\ell_p$}{lp} minimization and thresholding}
Motivated by the $\ell_p$ objective and its theoretical guarantee, we propose an iterative $\ell_p$ minimization algorithm. The algorithm, which is detailed in Algorithm \ref{alg:1}, alternates between updating the outlier indicator vector $\bm h$ via minimizing its $\ell_p$-norm and updating the estimated mean $\bm x$. To describe Algorithm~\ref{alg:1}, let $\mathcal{H}$ be the set defined by
\begin{align}
    \mathcal{H}(\bm x, c_2)\coloneqq &\arg\min_{\bm h}  \|\bm h\|_p\label{eq:Hset} \\
    s.t.\quad  & 0\leq h_i \leq 1, \forall i,\nonumber\\
     &\lambda_{\mathrm{max}}\left(\sum_{i=1}^{n} (1-h_i)(\bm y_i-\bm x)(\bm y_i-\bm x)^\top\right)\leq (c_1^2+c_2^2)\sigma^2 n .\nonumber
\end{align}

\begin{algorithm}[tb]
   \caption{Robust Mean Estimation via $\ell_p$ Minimization and Thresholding}
   \label{alg:1}
\begin{algorithmic}
   \STATE {\bfseries Inputs:}\\
   1) An $\epsilon$-corrupted set of datapoints $\{\bm y_i\}_{i=1}^{n}\in \mathbb{R}^d$ generated by a distribution whose covariance matrix satisfies $\Sigma\preceq\sigma^2 I$.\\ 2) Upper bound on corruption level: $\check{\epsilon}$\\
   3) Upper bound on spectral norm of $\Sigma$: $\sigma^2$.\\
   4) Threshold: $0<\tau\leq 1$ such that $f(\tau)>\check{\epsilon}$, where $f(\tau)$ is defined in \eqref{eq:f}, if such a $\tau$ exists.\\
   5) Set $c_1>1$.\\
   6) Set $0<p\leq 1$ in $\ell_p$.
   \STATE {\bfseries Initialize:}\\
   1) $\bm x^{(0)}$ as the coordinate-wise median of $\{\bm y_i\}_{i=1}^n$.\\
   2) $c_2^{(0)}=3\sqrt{d}+2c_1$.\\
   3) Iteration number $t=0$.
   \STATE {\bfseries Do:}
  \STATE \emph{Step 1}: Given $\bm x^{(t)}$, update $\bm h$:
   \STATE $\bm h^{(t)}\in \mathcal{H}(\bm x^{(t)},c_2^{(t)})$, 
     where $\mathcal{H}$ is defined in \eqref{eq:Hset}.
    \STATE \emph{Step 2}: Given $\bm h^{(t)}$, update $\bm x$:\\
   $\bm x^{(t+1)}=\frac{\sum_{i=1}^{n} (1-h_i^{(t)})1\{h_i^{(t)}\leq\tau \}\bm y_i}{\sum_{i=1}^{n} (1-h_i^{(t)})1\{h_i^{(t)}\leq\tau \}}$.\\
   $c_2^{(t+1)}=\gamma(\check{\epsilon})c_2^{(t)}+\beta(\check{\epsilon})$,\\
   where $\gamma$ and $\beta$ are defined in \eqref{eq:gamma} and \eqref{eq:beta}\\
   $t=t+1$.\\
   \STATE {\bfseries While:} $t< T = 1+\frac{\log c_2^{(0)}}{\log\abs{\gamma(\check{\epsilon})}}$ and {$c_2^{(t)}< c_2^{(t-1)}$}
   \STATE {\bfseries Output:} $\bm x^{(T)}$
\end{algorithmic}
\end{algorithm}

 When updating the estimated mean $\bm x$ in Step 2 of Algorithm \ref{alg:1}, we add an option to threshold the $h_i$ by $\tau$, so one can use the weighted average of the estimated `reliable' datapoints (i.e., those for which $h_i\approx 0$) to estimate $\bm x$. This is motivated by the analysis of the original $\ell_0$ objective in Theorem~\ref{thm:l0}, where the average of the estimated `reliable' datapoints $\frac{\sum_{\{i:\hat{h}_i=0\}} \bm y_i}{|\{i:\hat{h}_i=0\}|}$ is close to the true mean as long as the outlier indicator vector $\hat {\bm h}$ is sparse enough. The breakdown point of Algorithm \ref{alg:1} depends on the threshold $\tau$ and is given by $f(\tau)$ (see \eqref{eq:f}). The maximal breakdown point corresponds to no thresholding, i.e., $f(1)=1-1/\sqrt{2}$. Algorithm \ref{alg:1} requires an upper bound $\check{\epsilon}$ on the true fraction of outliers. This upper bound can be set arbitrarily close to (but less than) the breakdown point.
 
 With this intuitive updating rule in Step 2, Algorithm \ref{alg:1} has following order-optimal guarantee.


\begin{theorem}\label{thm:algo}
Let $P$ be a distribution on $\mathbb{R}^d$ with unknown mean $\bm\mu$ and unknown covariance matrix $\Sigma\preceq \sigma^2 I$. Let $\delta\in(0,1/5)$, $c_1>1$ and $p\in(0,1]$ be fixed. Let $\tau\in(0,1]$, {$c_1'=c^2_1 \min\bre*{c^2_1\log c^2_1+1-c^2_1,1}$, $n>\max\bre*{90,\frac{e}{c_1'\delta^2f(\tau)}d}\log(d/\delta)$, $\alpha=\frac{ed\log(d/\delta)}{n\delta^2c'_1}$}. Let $\epsilon\geq 0$ be such that $0< \epsilon'\coloneqq \epsilon +\alpha \leq \check{\epsilon}< f(\tau)$. Given an $\epsilon$-fraction corrupted set of $n$ datapoints from $P$, with probability at least $1-5\delta$, {all} the iterates of Algorithm~\ref{alg:1} (for $t\geq 1$) satisfy 
\begin{equation}
    \begin{split}
    \|{\bm x}^{(t)}-{ {\bm\mu}}\|_2 \leq\, & \sigma\left[\gamma(\epsilon')\para*{c_2^{(0)}\gamma(\check{\epsilon})^{t-1}+\frac{1-\gamma(\check{\epsilon})^{t-1}}{1-\gamma(\check{\epsilon})}\beta(\check{\epsilon})}+\beta(\epsilon')\right]+c_1\sigma\sqrt{\frac{\epsilon}{(1-\alpha)(1-\epsilon)}}\\
        &+{\sigma\sqrt{\alpha\delta}\para*{1+2\sqrt{\frac{c_1'}{e\log(d/\delta)}}}}
    \end{split}
\end{equation}

where $c_2^{(0)}$ is given in Algorithm~\ref{alg:1}, and
\begin{align}
    f(\tau) &= \frac{3\tau+\tau^2-\sqrt{\tau^4+2\tau^3+5\tau^2}}{2(1+\tau)}\label{eq:f}\\
     \gamma(\epsilon) &= \sqrt{\frac{\epsilon/\tau}{(1-\epsilon/\tau)(1-\epsilon-\epsilon/\tau)}}\label{eq:gamma}\\
    \beta(\epsilon)&=c_1\para*{(1-\epsilon/\tau)^{-1/2}+(1-\epsilon)^{-1/2}}\sqrt{\frac{\epsilon/\tau}{1-\epsilon-\epsilon/\tau}} \label{eq:beta}.
\end{align}
The output of Algorithm \ref{alg:1} at the end of $T = 1+\frac{\log c_2^{(0)}}{|\log \gamma(\check{\epsilon})|}=O\para*{\frac{\log d}{|\log \check{\epsilon}|}}$ (when $c_2^{(0)}\geq \frac{\beta(\check{\epsilon})}{1-\gamma(\check{\epsilon})}$) or $T=1$ (when $c_2^{(0)}< \frac{\beta(\check{\epsilon})}{1-\gamma(\check{\epsilon})}$) iterations is order-optimal:
\begin{equation}
\begin{split}
    \|\bm x^{(T)}-\bm\mu\|_2 &\leq \sigma\left[\gamma(\epsilon')\para*{1+\frac{\beta(\check{\epsilon})}{1-\gamma(\check{\epsilon})}}+\beta(\epsilon')\right]+c_1\sigma\sqrt{\frac{\epsilon}{(1-\alpha)(1-\epsilon)}}+{\sigma\sqrt{\alpha\delta}\para*{1+2\sqrt{\frac{c_1'}{e\log(d/\delta)}}}}\\
    &= O(\sigma\sqrt{\epsilon'}).
\end{split}
\end{equation}
\end{theorem}

The proof is deferred to the Appendix, but we briefly discuss the design of the algorithm and the high-level approach. Let $\bar{\bm x}^*$ be the average of the set of inlier datapoints that are within a distance of $\sigma\sqrt{\frac{d}{\alpha\delta}}$ from $\bm\mu$. We use induction to show that $\|\bm x^{(t)} -\bar{\bm x}^*\|\leq c_2^{(t)}\sigma$. We show in the Appendix that the coordinate-wise median satisfies $\|\bm x^{(0)} -\bm\mu\|_2\leq c_2^{(0)}\sigma$ with high probability. Firstly,  observe that in Step 1 of Algorithm~\ref{alg:1}, the constraint on the spectral norm of the weighted covariance matrix around $\bm x^{(t)}$ is $\left(c_1^2+(c_2^{(t)})^2\right)\sigma^2n$ instead of $c_1^2\sigma^2n$ as in \eqref{objL1}. This ensures that with high probability that the optimization problem in Step 1 has a feasible point, and that the optimum solution satisfies $\|\bm h^{(t)}\|_p\leq (\epsilon'n)^{1/p}$. Secondly, we exploit the boundedness of $\|\bm h^{(t)}\|_p$ and the fact that the spectral norm of the weighted covariance matrix around $\bm x^{(t)}$ is bounded (similar to the idea used in Theorem 2), along with some concentration bounds to show that in each iteration the iterate $\bm x^{(t+1)}$ in Step 2 moves closer to $\mu$ than $\bm x^{(t)}$. Specifically, we show that $\|\bm x^{(t+1)}-\bm\mu\|_2\leq \gamma\|\bm x^{(t)} -\bm\mu\|_2+\beta\sigma\leq (\gamma c_2^{(t)}+\beta)\sigma=c_2^{(t+1)}\sigma $, where $\gamma<1$. From the proof we can see that it is not necessary to reach the global optimum in Step 1, we only need to find a feasible point whose $\ell_p$-norm is small enough.


\begin{remark}
The results of Theorems \ref{thm:l0}, \ref{thm:lp2} and \ref{thm:algo} can be easily extended to establish the estimators' closeness to the average of the datapoints before corruption, $\tilde{\bm\mu}=\frac{1}{n}\sum\limits_{i=1}^n\tilde{\bm y}_i$,  using the fact that $\tilde{\bm\mu}$ is close to $\bm\mu$, which is shown in the Appendix $($see \eqref{event:3}$)$. We obtain the following extension to the above theorems with the same probability guarantees:
\begin{align}
    \|\hat{\bm x}-\tilde{\bm\mu}\|_2\leq \|\hat{\bm x}- \bm\mu\|_2 + {\sigma\sqrt{\frac{d}{n\delta}}}.
\end{align}
Moreover, it can be also shown that the estimators are close to the average of inliers, that are at most a distance of {$\sigma\sqrt{\frac{d}{\alpha\delta}}=\sigma\sqrt{\frac{n\delta c'_1}{e\log(d/\delta)}}$} from $\bm\mu$.
\end{remark}

\begin{remark}
The initialization $c_2^{(0)}=3\sqrt{d}+2c_1$ can be replaced by a smaller value as long as it is possible to guarantee $\|\bm x^{(0)}-\bm\mu\|_2\leq c_2^{(0)}\sigma$ with high probability.
\end{remark}

An important aspect of the proposed algorithm is that it does not require the true fraction of outliers $\epsilon$ and is still order-optimal. To the best of our knowledge no other algorithm for our corruption model has this property.


\subsection{Solving Step 1 of Algorithm~\ref{alg:1}}
\label{headings}

When we set $p=1$ in the objective $\|\bm h\|_p$ in Step~1 of Algorithm~\ref{alg:1}, the resulting problem is convex, and can be reformulated as the following packing SDP ~\cite{iyengar2005approximation} with $w_i\triangleq 1-h_i$, and $e_i$ being the $i$-th standard basis vector in $\mathbb{R}^n$. The details can be found in the Appendix.
\begin{align}\label{packing SDP}
   \max_{\bm w} \ \bm 1^\top \bm w \quad 
    s.t. &\  w_i \geq 0, \forall i\\
   & \sum_{i=1}^{n} w_i \begin{bmatrix}
    e_ie_i^\top &  \\
    &  (\bm y_i-\bm x)(\bm y_i-\bm x)^\top \end{bmatrix} \preceq \begin{bmatrix}
    I_{n\times n} &  \\
    &  c n \sigma^2 I_{d\times d} \end{bmatrix} \nonumber
\end{align}

When $0<p<1$, the equivalent objective function $\|\bm h\|_p^p = \sum_i h_i^p$ is concave, \textit{not }convex. So it may be difficult to find its global minimum. Nevertheless, we can iteratively construct and minimize a \textit{tight} upper bound on this objective function via iterative re-weighted $\ell_2$~\cite{4518498,gorodnitsky1997sparse} or  $\ell_1$ techniques~\cite{candes2008enhancing} from compressive sensing.\footnote{We observe that iterative re-weighted $\ell_2$ achieves better empirical performance.} And it is well-known in compressive sensing that such iterative re-weighted approaches often performs better than $\ell_1$~\cite{candes2008enhancing,4518498}.

\subsection{Complexity analysis}

Theorem~\ref{thm:algo} guarantees that the total number of iterations of Algorithm \ref{alg:1} required to achieve optimality is upper bounded by $O(\frac{\log d}{\log\abs*{\check{\epsilon}}})$. 
In each iteration, the computational complexity of Step 2 is $O (nd)$. It follows easily from the proof of Theorem \ref{thm:algo}, that it suffices to solve the SDP in step 1 of Algorithm 1 (with $p=1$) to a \textit{constant} precision. As a result, the error is affected by a constant and thus remains order-optimal and the time complexity is $\tilde O(nd)$ parallelizable work using positive SDP solvers~\cite{Allen-Zhu2016} (the notation $\tilde{O}(m)$ hides the poly-log factors: $\tilde{O}(m)=O(m.\text{polylog}(m))$). {A comparison of our theoretical results with those in state-of-the-art works is given in Table \ref{table:theory}.}

If we use $\ell_p$ with $0<p<1$ in Step 1, we iteratively construct and minimize a \textit{tight} upper bound on the $\ell_p$ objective via iterative re-weighted $\ell_2$~\cite{4518498,gorodnitsky1997sparse} or iterative re-weighted $\ell_1$ techniques~\cite{candes2008enhancing}\footnote{We run fewer than 10 re-weighted iterations in our implementation.}. Minimizing the resulting weighted $\ell_1$ objective can be also solved very efficiently to a \emph{constant} precision by formulating it as a Packing SDP (see 
Appendix) with computational complexity of $\tilde{O}(nd)$~\cite{Allen-Zhu2016}. If we use iterative re-weighted $\ell_2$, minimizing the resulting weighted $\ell_2$ objective is \RED{a SDP constrained least squares problem, whose computational complexity is in general polynomial in both $d$ and $n$.} We will explore more  efficient solutions for this objective in future work.

\begin{table}[!ht]
    \centering
    \begin{tabular}{|c|c|c|c|c|}
    \hline
         Algorithm & Time complexity & Error guarantee & Breakdown point & Requires $\epsilon$  \\
         \hline
         Tukey median~\cite{zhu2020does} & NP-hard & $O(\sigma\sqrt{\epsilon})$ & $\frac{1}{d+1}$ & No\\
         \hline
         IF~\cite{diakonikolas2017being} & $\tilde{O}(nd^2)$ & $O(\sigma\sqrt{\epsilon})$ & $NA$ & Yes\\
         \hline
         GF~\cite{zhu2020robust} & $\tilde{O}(n^2d)$ & $O\para*{\sigma\frac{\sqrt{\epsilon}}{1-2\epsilon}}$ & $\frac{1}{2}$ & Yes\\
         \hline
         CDG~\cite{convex} & $\tilde{O}\para*{\frac{nd}{\epsilon^6}}$ & $O(\sigma\sqrt{\epsilon})$ & $\frac{1}{3}$ & Yes\\
         \hline
         QUE~\cite{NIPS2019_8839} & $\tilde{O}(nd)$ & $O(\sigma\sqrt{\epsilon})$ & NA & Yes\\
         \hline
         Proposed optimization problems $(\ell_p, p\in[0,1])$ & NA & $O\para*{\sigma\sqrt{\frac{\epsilon}{1-2\epsilon}}}$ & $\frac{1}{2}$ & No\\
         \hline
         Proposed algorithm $(p=1)$ & $\tilde{O}(nd)$ & $O(\sigma\sqrt{\epsilon})$ & $1-\frac{1}{\sqrt{2}}\approx 0.3$ & No\\
         \hline
    \end{tabular}
    \caption{Theoretical comparison}
    \label{table:theory}
\end{table}

\section{Empirical Studies}
\label{others}

In this section, we present empirical results on the performance of Algorithm~\ref{alg:1} and compare with the following state-of-the-art high dimension robust mean estimation methods: Iterative Filtering (IF)~\cite{diakonikolas2017being}, Generalized Filtering (GF)~\cite[Algorithm 2]{zhu2020robust}, the method proposed in~\cite{LRV} (denoted as LRV), the method {for bounded covariance distributions} in~\cite{convex} (denoted as CDG), and Quantum Entropy Scoring (QUE)~\cite{NIPS2019_8839}, which scores the outliers based on multiple directions.  {We briefly discuss the implementation details of these algorithms. We implemented the QUE method by utilizing the code provided in \cite{NIPS2019_8839}. In \cite{convex}, the authors provide a way to implement the CDG method approximately; we provide results for an exact implementation of the CDG method. Since the number of datapoints considered in the following simulations is less than the minimum requirement, a reasonable approach to compare the performance of the algorithms is to tune the hyper-parameters of all algorithms to get the best possible error. For example, the hyper-parameter $c_4$ that appears in CDG method~\cite{convex}, is set to be $1.05$, which produced the smallest empirical error. The value of $\sigma$ provided to the algorithms is not the theoretical value, but the empirical one (precisely, the spectral norm of the sample covariance matrix of $G^*$ (see Section \ref{sec:propose}))}. For evaluation purposes, we report the recovery error, which we define it as the $\ell_2$ distance of the estimated mean to the oracle solution, i.e., the average of the uncorrupted datapoints after corruption.
\subsection{Synthetic data}\label{synthetic data}

We consider two experimental settings. For the first setting we follow~\cite{NIPS2019_8839}. The dimension of the data is $d$, and the number of datapoints is $n$. The inlier datapoints are generated i.i.d. according to the standard Gaussian distribution with zero mean. Randomly (uniformly) chosen $\epsilon$ fraction of the datapoints are replaced by outliers. For the outliers, half of them are set to be $(\sqrt{d/2},\sqrt{d/2},0,...,0)^\top$, and the other half are set as $(\sqrt{d/2},-\sqrt{d/2},0,...,0)^\top$, so that their $\ell_2$ distances to the population mean $(0,...,0)^\top$ are all $\sqrt{d}$, similar to that of the inlier points. These are two clusters of outliers, and their $\ell_2$ distances to the true mean $\bm x$ are similar to that of the inlier points.  
In Algorithm \ref{alg:1}, we set the threshold $\tau=0.6$, $c_1=1.1$, and we initialize $c_2^{(0)}$ as the $\ell_2$ error of the coordinate-wise Median relative to the true mean. {We implemented the IF method for sub-Gaussian parameters~\cite[Theorem 3.1]{diakonikolas2017being}.} We vary the total fraction $\epsilon$ of the outliers and report the average recovery error of each method over 10 trials in Table~\ref{settingB_100} with $d=100, n=1000$. The proposed $\ell_1$ and $\ell_{0.5}$ methods show significant improvements over the competing methods, and the $\ell_{0.5}$ method performs the best.
%
\begin{table}[!ht]\centering
	\caption{Recovery error of each method under different fraction $\epsilon$ of the outlier points ($d=100, n=1000$)}
	\label{settingB_100}
	\begin{tabular}{|c|c|c|c|c|c|c|c|c|c|}
		
		\hline
		$\epsilon$	   & IF & GF & QUE & LRV & CDG & $\ell_1$ & $\ell_{0.5}$\\
		\hline
		10\% & 0.124 &0.098 & 0.429 & 0.367 & 0.064 &\textbf{0.013}& \textbf{0.006} \\
		\hline
		20\% & 0.131&0.115 & 0.492 & 0.659 & 0.084 &\textbf{0.013}& \textbf{0.007} \\
		\hline
	\end{tabular}
\end{table}

We also tested the performance of each method for different numbers of datapoints. The dimension of the data is fixed to be 100. The fraction of the corrupted points is fixed to be 20\%. We vary the number of datapoints from 100 to 1000, and report the average recovery error for each method over 50 trials in Table~\ref{settingB_100_02}. We can see that the performance of all methods get better when the number of datapoints is increased. Again, our proposed methods consistently perform better than the other methods.

\begin{table}[!ht]\centering
	\caption{Recovery error of each method w.r.t. different number of samples ($d=100, \epsilon=0.2$)}
	\label{settingB_100_02}
	\begin{tabular}{|c|c|c|c|c|c|c|c|}
		
		\hline
		$n$ &  IF & GF & QUE & LRV & CDG & $\ell_1$ & $\ell_{0.5}$\\
		\hline
		100 & 0.493&0.293 & 1.547 & 1.423 & 0.316 & \textbf{0.060}&\textbf{0.033}\\
		\hline
		200 & 0.313 & 0.239&1.038 & 1.084 & 0.198 & \textbf{0.036}&\textbf{0.021} \\
		\hline
		500 & 0.186&0.170 & 0.680 & 0.794 & 0.148 & \textbf{0.021}&\textbf{0.012} \\
		\hline
		1000 & 0.131&0.115 & 0.492 & 0.659 & 0.084 &\textbf{0.013}& \textbf{0.007} \\
		\hline
	\end{tabular}
\end{table}

{The second experimental setting is as follows. The dimension of the data is $d$, and the number of datapoints is $n$. The inlier datapoints are generated i.i.d. such that each coordinate follows the Pareto distribution with scale parameter as $1$ and shape parameter as $2.5$. This implies that each coordinate has bounded second moment but no higher moments. All outliers are set to be the same vector $\bm v$ which is chosen as follows. Let $g$ be the average of the $\ell_2$ norms of the datapoints. The vector $\bm v$ is set as $(2+\sqrt{g/d},2+\sqrt{g/d},\dots,2+\sqrt{g/d})^\top$. Randomly (uniformly) chosen $\epsilon$ fraction of the datapoints are replaced by outliers. The LRV method is applicable only to cases where distributions have bounded fourth moment, and hence is not applicable to this setting. The CDG method was not implemented due to the high computational complexity of its implementation. We implemented the IF method for bounded second moment parameters ~\cite[Theorem 3.2]{diakonikolas2017being}. We implemented Algorithm \ref{alg:1} with $p=1$, $\tau=1$(no thresholding during the iterations), $c_1=1$ and $c_2^{(0)}=3\sqrt{d}+2c_1$. However, we threshold the last iterate $h^{(T)}$ with threshold $0.6$. We vary the total fraction $\epsilon$ of the outliers and report the average recovery error of each method over 100 trials in Table~\ref{table:pareto1} with $d=1000, n=100000$. The proposed $\ell_1$ method show significant improvement over the competing methods. We also tested the performance of each method for two different numbers of datapoints. The dimension of the data is fixed to be 1000. The fraction of the corrupted points is fixed to be 20\%. We consider the number of datapoints to be 10000 and 100000, and report the average recovery error for each method over 100 trials in Table~\ref{table:pareto2}. Again, the proposed method consistently perform better than the other methods.} 

\begin{table}[!ht]\centering
	\caption{Recovery error of each method under different fraction $\epsilon$ of the outlier points ($d=1000, n=100000$)}
	\label{table:pareto1}
	\begin{tabular}{|c|c|c|c|c|}
		
		\hline
		$\epsilon$	   & IF & GF & QUE  & $\ell_1$\\
		\hline
		10\% & 0.0164 & 0.0550 & 0.1096 &\textbf{0.0161}\\
		\hline
		20\% & 0.0845 & 0.0726 & 0.2489 &\textbf{0.0190} \\
		\hline
	\end{tabular}
\end{table}

\begin{table}[!ht]\centering
	\caption{Recovery error of each method w.r.t. different number of samples ($d=1000, \epsilon=0.2$)}
	\label{table:pareto2}
	\begin{tabular}{|c|c|c|c|c|}
		
		\hline
		$n$ &  IF & GF & QUE & $\ell_1$\\
		\hline
		10000 & 0.3027 & 0.0780 & 0.3991 & \textbf{0.0257}\\
		\hline
		100000 & 0.0845 & 0.0726 & 0.2489 &\textbf{0.0190}\\
		\hline
	\end{tabular}
	
\end{table}

\subsection{Corrupted image dataset}

Here we use a dataset of real face images to test the effectiveness of the robust mean estimation methods. The average face of particular regions or certain groups of people is useful for many social and psychological studies~\cite{little2011facial}. Here we use 100 frontal human face images from the Brazilian face database\footnote{https://fei.edu.br/~cet/facedatabase.html} as inliers. For the outliers, we choose 15 face images of cats and dogs from the CIFAR10~\cite{krizhevsky2009learning} database. In order to be able to run the CDG method~\cite{convex}, we scale the size of images to 18 $\times$ 15 pixels, so the dimension of each datapoint is 270. The oracle solution is the average of the 100 human faces. Table~\ref{face} reports the recovery error, which is the $\ell_2$ distance of the estimated mean to the oracle solution, for each method. The proposed methods achieve smaller recovery error than the state-of-the-art methods. The sample inlier and outlier images as well as the estimated mean for each method can be found in the Appendix. 

\begin{table}[!ht]\centering
	\caption{Recovery error of the mean face by each method}
	\label{face}
	\begin{tabular}{|c|c|c|c|c|c|c|}
		
		\hline
		Sample average  & IF & LRV & CDG & $\ell_1$ & $\ell_{0.5}$ \\
		\hline
		141  & 63 & 83 & 81 & \textbf{38} & \textbf{46} \\
		\hline

	\end{tabular}
\end{table}

\section{Conclusion} 


We formulated the robust mean estimation problems as the minimization of  the  $\ell_0$-`norm' of the introduced {outlier indicator vector}, under a second moment constraint on the datapoints. We further relaxed the $\ell_0$ objective to an $\ell_p$ $(0<p\leq 1)$ objective, and theoretically justified the new objective. The proposed $\ell_0$ and $\ell_p$ optimization problems do not need to know $\epsilon$, and still achieve information-theoretically order-optimal error bounds with optimal breakdown points. Then we proposed a computationally tractable iterative $\ell_p (0<p\leq 1)$ minimization and hard thresholding algorithm, which significantly outperforms state-of-the-art robust mean estimation methods, and is order-optimal. In the empirical studies, we observed strong numerical evidence that using the $\ell_p$ $(0<p\leq 1)$ norm in the optimization leads to sparse solutions; theoretically justifying this phenomenon is also of interest. It is worth noting that almost all previous polynomial-time methods (with dimension-independent error bound) need to know $\epsilon$, while our Algorithm 1 does not require to know $\epsilon$. It has a maximal breakdown point of $1-1/\sqrt{2}$, and has near-linear time complexity for $p=1$.  
\bibliographystyle{IEEEtran}

\bibliography{main}
\newpage
\section{Appendix}

\subsection{Technical preliminaries}
We introduce the following parameters that control the minimum number of datapoints needed, error and confidence level. Let $\delta>0$, $c_1>1$ and $c_1'=[c^2_1 \min\bre*{c^2_1\log c^2_1+1-c^2_1,1}]$. {Let $n>\frac{ed\log(d/\delta)}{\delta^2c'_1}$ and $\alpha=\frac{ed\log(d/\delta)}{n\delta^2c'_1}$.} Let $\bm S=\lbrace\tilde{\bm y}_1,\dots,\tilde{\bm y}_n\rbrace$ be a set of $n$ datapoints drawn from a distribution $P$ with mean $\bm\mu$ and covariance matrix $\Sigma\preceq \sigma^2 I$. We now define $\bm G $ as the set of datapoints which are less than {$\sigma\sqrt{\frac{d}{\alpha\delta}}=\sigma\sqrt{\frac{n\delta c'_1}{e\log(d/\delta)}}$} distance away from $\bm\mu$: 
\begin{align}
\bm I &=\left\lbrace i: \lVert \tilde{\bm y}_i-\bm\mu \rVert_2\leq \sigma\sqrt{\frac{d}{\alpha\delta}} \right\rbrace\label{eq:I1}\\
\bm G &=\lbrace \tilde{\bm y}_i: i\in\bm I\rbrace.\label{eq:Gdef}    
\end{align}
It follows from Lemma \ref{lem:size} that for the event
\begin{align}
    &\mathcal{E}_1 = \{|\bm I|\geq n-\alpha n\}\label{event:1},\\
    &\pr(\mathcal{E}_1) \geq 1-\delta\label{eq:prob1}.
\end{align}
Let $\mathcal{E}_2$ be the event:

\begin{equation}\label{event}
    \mathcal{E}_2 = \left\lbrace \lambda_{\mathrm{max}}\left(\sum_{i\in \bm I} (\tilde{\bm y}_i-\bm\mu)(\tilde{\bm y}_i-\bm\mu)^\top\right)\leq c^2_1   \sigma^2n \right\rbrace.
\end{equation}
It follows from Lemma \ref{lem:event2} that
\begin{align}
    \pr(\mathcal{E}_2)\geq 1-\delta.
\end{align}

Thus, we have that 
\begin{equation}\label{eq:existence}
    \pr(\mathcal{E}_1\cap \mathcal{E}_2) \geq 1-2\delta.
\end{equation}
For analysis purposes, we consider the \textit{far away} uncorrupted datapoints $\bm S\setminus \bm G$ as outliers also.


 Let $\lbrace \bm y_1,\dots,\bm y_n \rbrace$ be an $\epsilon$-corrupted version of the set $\bm S$. Let $\bm h^*$ be such that $h_i^*=1$ for the outliers (both far away uncorrupted datapoints and corrupted datapoints), and $h_i^*=0$ for the rest of uncorrupted datapoints, i.e.,
 \begin{equation}\label{eq:outlier}
 h_i^*=
     \begin{cases}
     1, &\text{ if\quad $\bm y_i\neq \tilde{\bm y}_i$\quad or\quad $\tilde{\bm y}_i\in \bm S\setminus \bm G$}\\
     0, &\text{otherwise}
     \end{cases}
 \end{equation}

Let the set of \textit{inliers} be given by $\bm G^*$:
\begin{align}
    \bm{I}^* &=\{i:h_i^*=0\}\label{eq:I2}\\
    \bm{G}^* &=\lbrace \bm y_i: i\in\bm{I}^*\rbrace=\lbrace \tilde{\bm y}_i: i\in\bm{I}^*\rbrace\label{eq:G2}
\end{align}

Note that $\bm I^*\subseteq \bm I$ and $\bm G^*\subseteq \bm G$. 
Since $(\tilde{\bm y}_i-\bm\mu)( \tilde{\bm y}_i-\bm\mu)^\top$ is positive semi-definite (PSD), we must have 
\begin{align*}
    \lambda_{\mathrm{max}}\left(\sum_{i=1}^{n} (1-h_i^*)(\bm y_i-\bm\mu)(\bm y_i-\bm\mu)^\top\right)\leq \lambda_{\mathrm{max}}\left(\sum_{i\in \bm I} (\tilde{\bm y}_i-\bm\mu)(\tilde{\bm y}_i-\bm\mu)^\top\right).
\end{align*}
This implies that
\begin{equation}\label{event2}
    \left\lbrace\lambda_{\mathrm{max}}\left(\sum_{i=1}^{n} (1-h_i^*)(\bm y_i-\bm\mu)(\bm y_i-\bm\mu)^\top\right)\leq c_1^2  \sigma^2n\right\rbrace\supseteq \mathcal{E}_2.
\end{equation}
Then, we have:
\begin{equation}\label{goodPoints_cov}
\pr\left\lbrace\lambda_{\mathrm{max}}\left(\sum_{i=1}^{n} (1-h_i^*)(\bm y_i-\bm\mu)(\bm y_i-\bm\mu)^\top\right)\leq c_1^2  \sigma^2n\right\rbrace
\geq\pr(\mathcal{E}_2) \geq 1-\delta.
\end{equation}
Our intended solution is to have $h_i=0$ for the inlier points and $h_i=1$ for the outlier points.

Let $\bar{\bm x}$ and $\bar{\bm x}^*$ be the averages of datapoints in $\bm G$ and $\bm G^*$ respectively. Applying Lemma C.2 from \cite{zhu2020robust}, we have
\begin{equation}\label{eq:res}
    \|\bar{\bm x}-\bar{\bm x}^*\|_2\leq \sqrt{\frac{c_1^2\sigma^2}{1-\alpha}.\frac{\epsilon}{1-\epsilon}}.
\end{equation}

We now introduce some more events (c.f.  \cite[Lemma A.18]{diakonikolas2017being}):
\begin{align}
    \mathcal{E}_3 &= \bre*{\left\lVert \frac{1}{n}\sum\limits_{i=1}^n (\tilde{\bm y}_i-\bm\mu) \right\rVert_2\leq {\sigma\sqrt{\frac{d}{n\delta}}}}\label{event:3}\\
    \mathcal{E}_4 &= \bre*{\left\lVert \frac{1}{n}\sum\limits_{i=1}^n (\bm z_i-\E[\bm z_1]) \right\rVert_2\leq {\sigma\sqrt{\frac{d}{n\delta}}}}\label{event:4},
\end{align}
where $\bm z_i=(\tilde{\bm y}_i - \bm\mu)\mathbbm{1}\bre*{ \lVert \tilde{\bm y}_i-\bm\mu \rVert_2> \sigma\sqrt{\frac{d}{\alpha\delta}}}$. From Lemma \ref{lem:size}, we get that
\begin{align}
    \pr(\mathcal{E}_3)\geq 1-\delta, \text{ and } \pr(\mathcal{E}_4)\geq 1-\delta\label{eq:event34}.
\end{align}
Let $\mathcal{E}$ be the event given by
\begin{align}\label{event:all}
    \mathcal{E}=\mathcal{E}_1\cap \mathcal{E}_2\cap \mathcal{E}_3\cap \mathcal{E}_4.
\end{align}
Let $\Delta_{n,\xi}$ be the set of probability vectors given by:
\begin{equation}\label{eq:Delta}
    \Delta_{n,\xi}=\bre*{\bm w \in\mathbb{R}^n: 0\leq w_i\leq\frac{1}{1-\xi}, \sum\limits_{i=1}^n w_i=1}.
\end{equation}
Let $\mathrm{TV}(.,.)$ denote the total variation distance between probability measures.
\begin{table}[!ht]
    \centering
    \caption{Description of variables}
    \begin{tabular}{|c|c|}
          \hline    
          Variable &  Description\\
          \hline\hline
          $\bm \mu$ & Mean (expected value) of population distribution\\
          \hline
          $\tilde{\bm\mu}$ & Average of all datapoints before corruption\\
          \hline
          $\bm G$ & Set of datapoints within {$\sigma\sqrt{\frac{d}{\alpha\delta}}=\sigma\sqrt{\frac{n\delta c'_1}{e\log(d/\delta)}}$} of  $\bm\mu$ before corruption\\
          \hline
          $\bm G^*$ & Maximal subset of $\bm G$ which is uncorrupted by adversary\\
          \hline
          $\bar{\bm x}$ & Average of vectors in $\bm G$, the set of datapoints within {$\sigma\sqrt{\frac{d}{\alpha\delta}}=\sigma\sqrt{\frac{n\delta c'_1}{e\log(d/\delta)}}$} of  $\bm\mu$\\
          \hline
          $\bar{\bm x}^*$ & Average of vectors in $\bm G^*$, the set of inliers within $\bm G$\\
          \hline
    \end{tabular}
    
    \label{tab:my_label}
\end{table}

\subsection{Technical Lemmas}\label{sec_lemma2}

\begin{lemma}[Lemma 2.2 \cite{zhu2020robust}]\label{lem:res} For a finite set of datapoints $\{\bm y_i\}_{i=1}^n$, let $\bm x_{\bm w} =\sum\limits_{i \in [n]} w_i\bm y_i$ and $\Sigma_{\bm w}=\sum\limits_{i \in [n]} w_i(\bm y_i-\bm x_{\bm w})(\bm y_i-\bm x_{\bm w})^\top$ be the weighted average and weighted covariance with respect to a probability weight vector $\bm w$. Let $\bm w_1$ and $\bm w_2$ be two probability weight vectors such that $\mathrm{TV}(\bm w_1,\bm w_2)\leq\zeta.$ Then,
\begin{align}
    \|\bm x_{\bm w_1}-\bm x_{\bm w_2}\|_2\leq \left(\sqrt{\lambda_{\max}(\Sigma_{\bm w_1})}+\sqrt{\lambda_{\max}(\Sigma_{\bm w_2})}\right)\sqrt{\frac{\zeta}{1-\zeta}}
\end{align}
\end{lemma}

\begin{lemma}[Lemma 2.3 \cite{zhu2020robust}]\label{lem:tv}
Let $\bm w_1\in\Delta_{n,\epsilon_1}$ and $\bm w_2 \in \Delta_{n,\epsilon_2}$. Then
\begin{align}
    \mathrm{TV}(\bm w_1,\bm w_2)\leq \frac{\max\{\epsilon_1,\epsilon_2\}}{1-\min\{\epsilon_1,\epsilon_2\}}.
\end{align}
\end{lemma}

\begin{lemma}\label{lem:median}
Let $P$ be a distribution on $\mathbb{R}^d$ with mean $\bm\mu$ and covariance matrix $\Sigma\preceq \sigma^2 I$. Let $\epsilon\leq 1/3$. Given an $\epsilon$-fraction corrupted set of $n$ datapoints from $P$, the coordinate-wise median of the corrupted set, $\hat{\bm x}$, satisfies with probability at least $1-d\exp(-n/90)$ that
\begin{align}
    \|\hat{\bm x}-\bm\mu\|_2\leq 3\sigma\sqrt{d}.
\end{align}
\end{lemma}
\begin{proof}
We first show that with high probability the error in each dimension is bounded by $3\sigma$. Fix a coordinate, and let $\tilde{y}_i$, $y_i$, $\mu$ and $\hat x$ be the component of $\tilde{\bm y}_i$, $\bm y_i$, $\bm\mu$ and $\hat{\bm x}$ respectively in that coordinate. By Markov's inequality, we have
\begin{align}
    \pr(|\tilde y_i-\mu_i|\geq 3\sigma)\leq 1/9.
\end{align}
Let $b_i=1\{|\tilde y_i-\mu_i|\geq 3\sigma\}$. By Chernoff's inequality, we obtain
\begin{align}
    \pr\left(\sum\limits_{i=1}^n b_i\geq n/6\right)\leq \exp\left(-\frac{(0.5)^2n}{9(2+0.5)}\right)= \exp(-n/90).
\end{align}

Thus with high probability more than five-sixth of the datapoints satisfy $|\tilde y_i-\mu_i|\leq 3\sigma$, which implies that even if $\epsilon\leq 1/3$ fraction of datapoints are corrupted, we would have
\begin{align}
    |\hat x-\mu|\leq 3\sigma.
\end{align}

Applying union bound, we get that with probability at least $1-d\exp(-n/90)$, the error in each dimension is bounded by $3\sigma$ and hence $\|\hat{\bm x}-\bm\mu\|_2\leq 3\sigma\sqrt{d}$ holds.

\end{proof}

\begin{lemma}\label{lem:size} Let $0<\delta\leq 1$. Let $\mathcal{E}_1$, $\mathcal{E}_3$ and $\mathcal{E}_4$ be the events as described in \eqref{event:1}, \eqref{event:3} and \eqref{event:4}. Then,
\begin{equation*}
\pr(\mathcal{E}_1) \geq 1-\delta,\; \pr(\mathcal{E}_3) \geq 1-\delta,\text{ and }\pr(\mathcal{E}_4) \geq 1-\delta,\\     
\end{equation*}
\end{lemma}
\begin{proof}
By Markov's inequality we have

\begin{align}
    \pr(|\bm G^c|>\alpha n)&\leq \frac{\E[|\bm G^c|]}{\alpha n}\\
    &=\frac{\E\left[\sum\limits_{i=1}^n 1\left\lbrace \lVert \tilde{\bm y}_i-\bm\mu \rVert_2> \sigma\sqrt{\frac{d}{\alpha\delta}} \right\rbrace\right]}{\alpha n}\\
    &=\frac{ \pr\left( \lVert \tilde{\bm y}_1-\bm\mu \rVert_2> \sigma\sqrt{\frac{d}{\alpha\delta}}\right)}{\alpha}.
\end{align}
Applying Markov's inequality again, we have
\begin{align}
    \pr\left( \lVert \tilde{\bm y}_1-\bm\mu \rVert_2> \sigma\sqrt{\frac{d}{\alpha\delta}}\right)&\leq \frac{\alpha\delta \E\left[ \lVert \tilde{\bm y}_1-\bm\mu \rVert^2_2\right]}{\sigma^2 d}\\
    &= \frac{\alpha \delta \mathrm{Tr}(\E[(\tilde{\bm y}_1-\bm\mu)(\tilde{\bm y}_1-\bm\mu)^\top])}{\sigma^2d}\\
    &\leq \frac{\alpha\delta\sigma^2d}{\sigma^2d}\\
    &=\alpha\delta.
\end{align}
Thus, we get
\begin{align}
    \pr(|\bm G^c|>\alpha n)&\leq \delta\\ \pr(|\bm G|\geq (1-\alpha) n)&\geq 1-\delta.
\end{align}
This proves the result for $\mathcal{E}_1$. Applying Markov's inequality again, we obtain
\begin{align}
    \pr\para*{\left\lVert \frac{1}{n}\sum\limits_{i=1}^n (\tilde{\bm y}_i-\bm\mu) \right\rVert_2\leq {\sigma\sqrt{\frac{d}{n\delta}}}} &\leq \frac{\E\bra*{\left\lVert \frac{1}{n}\sum\limits_{i=1}^n (\tilde{\bm y}_i-\bm\mu) \right\rVert^2_2}}{{\frac{\sigma^2 d}{n\delta}}}\\
    &= {\frac{n\delta}{\sigma^2 d}}\sum\limits_{k=1}^d \E\bra*{ (\tilde{\mu}_k-\mu_k) ^2}\\
    &\leq {\frac{n\delta}{\sigma^2 d}}.\frac{d\sigma^2}{n}\\
    &=\delta.
\end{align}
This proves the result for $\mathcal{E}_3$. By similar reasoning, the result for $\mathcal{E}_4$ follows.
\end{proof}

\begin{lemma}\label{lem:event2}
Let $0<\delta\leq 1$, $c_1>1$, $c_1'=[c^2_1 \min\bre*{c^2_1\log c^2_1+1-c^2_1,1}]$, {$n>\frac{ed\log(d/\delta)}{\delta^2c'_1}$ and $\alpha=\frac{ed\log(d/\delta)}{n\delta^2c'_1}$.} Let $\mathcal{E}_2$ be the event described in \eqref{event}. Then
\begin{equation*}
    \pr(\mathcal{E}_2)\geq 1-\delta.
\end{equation*}
\end{lemma}

\begin{proof}
We adopt the approach in ~\cite[Lemma A.18 (iv)]{diakonikolas2017being}. Lemma A.19 from \cite{diakonikolas2017being} states that the following: Let $\{X_i\}_{i=1}^n$ be $d\times d$ positive semi-definite random matrices such that $\lambda_{\max}(X_i)\leq L$ almost surely for all $i$. Let $S=\sum\limits_{i=1}^n X_i$ and $M=\lambda_{\max}\para*{\E[S]}$. Then, for any $\theta>0$,
\begin{align}\label{eq:lem5_1}
    \E\bra*{\lambda_{\max}\para*{S}}\leq (e^\theta-1)M/\theta+L\log (d)/\theta,
\end{align}
and for any $\eta>0$,
\begin{align}\label{eq:lem5_2}
    \pr\para*{\lambda_{\max}\para*{S}\geq (1+\eta)M}\leq d\para*{\frac{e^\eta}{(1+\eta)^{1+\eta}}}^{M/L}.
\end{align}
We apply this result by assigning $X_i=(\tilde{\bm y}_i-\bm\mu)(\tilde{\bm y}_i-\bm\mu)^\top \mathbbm{1}\bre*{\|\tilde{\bm y}_i-\bm\mu\|_2\leq \sigma\sqrt{\frac{d}{\alpha\delta}}}$. Note that $\lambda_{\max}(X_i)\leq L=\frac{\sigma^2 d}{\alpha\delta}$ for all $i\in [n]$, and $M\leq n\lambda_{\max}(\E[X_1])\leq n\sigma^2$.
We consider two mutually exclusive cases:

1) Suppose that $M<e^{-1}\delta c_1^2\sigma^2 n$. Applying  \eqref{eq:lem5_1} with $\theta=1$, we obtain
\begin{align}
    \E[\lambda_{\max} (S)]\leq (e-1)M+L\log d.
\end{align}
Applying Markov's inequality, we obtain
\begin{align}
    \pr(\lambda_{\max} (S)\geq  c_1^2\sigma^2n) &\leq \frac{\E[\lambda_{\max} (S)]}{c_1^2\sigma^2n}\\
    &\leq \frac{(e-1)\delta c_1^2\sigma^2n}{ec_1^2\sigma^2n}+\frac{\sigma^2d\log d}{\alpha\delta c_1^2\sigma^2n}\label{eq:lem5_3}\\
    &\leq \frac{(e-1)\delta}{e}+\frac{\delta}{e}\label{eq:lem5_4}\\
    &=\delta.
\end{align}
The inequality in \eqref{eq:lem5_3} follows from the assumption that $M<e^{-1}\delta c_1^2\sigma^2 n$ and the inequality in \eqref{eq:lem5_4} follows from the fact that {$\alpha=\frac{ed\log(d/\delta)}{n\delta^2c'_1}$ and $c'_1\leq c_1^2$.}

2) Suppose that $M\geq e^{-1}\delta c_1^2\sigma^2 n$. Applying \eqref{eq:lem5_2} with $\eta=c_1^2-1$, we obtain
\begin{align}
    \pr(\lambda_{\max} (S)\geq  c_1^2\sigma^2n) \leq\; &\pr(\lambda_{\max} (S)\geq  c_1^2 M)\label{eq:lem5_5}\\
    \leq\; & d\para*{\frac{e^{c_1^2-1}}{(c_1^2)^{c_1^2}}}^{\frac{\delta c_1^2\sigma^2n}{e}.\frac{\alpha\delta}{\sigma^2 d}}\label{eq:lem5_6}\\
    \leq\; & \delta\label{eq:lem5_7}.
\end{align}
The inequality in \eqref{eq:lem5_5} follows from the fact that $M\leq n\sigma^2$, the inequality in \eqref{eq:lem5_7} follows from the fact that $e^\alpha<(1+\alpha)^{1+\alpha}$ for any $\alpha>0$, and the fact that {$\alpha=\frac{ed\log(d/\delta)}{n\delta^2c'_1}$ and $c_1'=[c^2_1 \min\bre*{c^2_1\log c^2_1+1-c^2_1,1}]$.}
\end{proof}

\begin{lemma}\label{lem:wmin}
Given a set of points $\bm y_i\in \mathbb{R}^d$, $i=1,\dots,n$, then for any $\bm w\in \mathbb{R}^n$ we have
\begin{align}
    {\bm x_w} \triangleq \frac{\sum\limits_{i=1}^n w_i\bm y_i}{\|\bm w\|_1}\in \mathop{\arg\min_{\bm x}}\lambda_{\max} \left(\sum_{i=1}^n w_{i} (\bm y_i -\bm x)(\bm y_i - \bm x)^\top\right)\label{eq:wmin}
\end{align}
\end{lemma}

\begin{proof}
We have
\begin{align}
    \min_{\bm x}\lambda_{\max} \left(\sum_{i=1}^n w_{i} (\bm y_i -\bm x)(\bm y_i - \bm x)^\top\right) = &\min_{\bm x}\max_{\bm \nu: \|\bm\nu\|_2=1}\sum_{i=1}^n w_{i} \dotp*{\bm y_i -\bm x,\bm\nu}^2\\
    \geq &\max_{\bm \nu: \|\bm\nu\|_2=1}\min_{\bm x}\sum_{i=1}^n w_{i} \dotp*{\bm y_i -\bm x,\bm\nu}^2\label{Leq:minRHS}\\
    =&\max_{\bm \nu: \|\bm\nu\|_2=1} \sum_{i=1}^n w_{i} \dotp*{\bm y_i-{\bm x_w},\bm\nu}^2\label{eq:minRHS}\\
    =& \lambda_{\max}\left(\sum_{i=1}^n w_{i} (\bm y_i -{\bm x_w})(\bm y_i - {\bm x_w})^\top\right).
\end{align}
The equality \eqref{eq:minRHS} follows from the fact that the minimum in the RHS of \eqref{Leq:minRHS} is attained at ${\bm x_w} = \frac{\sum\limits_{i=1}^n w_i\bm y_i}{\|\bm w\|_1}$. Consequently, \eqref{eq:wmin} holds.

\end{proof}

\begin{lemma}\label{lem:feasible} Let {$n>\frac{ed\log(d/\delta)}{\delta^2c'_1}$ and $\alpha=\frac{ed\log(d/\delta)}{n\delta^2c'_1}$.}
Suppose $\|\bm x-\bar{\bm x}^*\|_2\leq c_2\sigma$, where $\bar{\bm x}^*$ is the average of datapoints in $\bm G^*$, defined in \eqref{eq:G2}. Then on event $\mathcal{E}_2$ defined in \eqref{event}, $\bm h^*$ satisfies
\begin{align}
    \lambda_{\max} \left(\sum_{i=1}^n (1-h^*_{i}) (\bm y_i -\bm x)(\bm y_i - \bm x)^\top\right)\leq (c_1^2+c_2^2)\sigma^2n.
\end{align}
\end{lemma}
\begin{proof}
 Let $\bm I$ and $\bm I^*$ be the sets defined in \eqref{eq:I1} and \eqref{eq:I2}. We have
\begin{align}
    &\lambda_{\mathrm{max}}\left(\sum_{i=1}^{n} (1-h_i^*)(\bm y_i-\bm x)(\bm y_i-\bm x)^\top\right)\\
    =&\lambda_{\mathrm{max}}\left(\sum_{i\in \bm I^*} ({\bm y}_i-\bm x)({\bm y}_i-\bm x)^\top\right) \\
    =&\lambda_{\mathrm{max}}\left(\sum_{i\in \bm I^*} ({\bm y}_i-\bar{\bm x}^*+\bar{\bm x}^*-\bm x)({\bm y}_i-\bar{\bm x}^*+\bar{\bm x}^*-\bm x)^\top\right)\\
    \leq &\lambda_{\mathrm{max}}\left(\sum_{i\in \bm I^*} ({\bm y}_i-\bar{\bm x}^*)({\bm y}_i-\bar{\bm x}^*)^\top\right)+\lambda_{\mathrm{max}}\left(\sum_{i\in \bm I^*} (\bm x-\bar{\bm x}^*)(\bm x-\bar{\bm x}^*)^\top\right)\\ &+2\lambda_{\mathrm{max}}\left(\sum_{i\in \bm I^*} ({\bm y}_i-\bar{\bm x}^*)(\bar{\bm x}^*-\bm x)^\top\right)\\
    =&\lambda_{\mathrm{max}}\left(\sum_{i\in \bm I^*} ({\bm y}_i-\bar{\bm x}^*)({\bm y}_i-\bar{\bm x}^*)^\top\right)+|\bm I^*| \|\bm x-\bar{\bm x}^*\|^2
    + 0\\
    \leq &c_1^2\sigma^2n+ c_2^2 \sigma^2n.
\end{align}
The last inequality follows from the definition of $\mathcal{E}_2$ in \eqref{event} and Lemma \ref{lem:wmin}.
\end{proof}

\begin{lemma}\label{lem:conc} Let {$n>\frac{ed\log(d/\delta)}{\delta^2c'_1}$ and $\alpha=\frac{ed\log(d/\delta)}{n\delta^2c'_1}$.}
Let $\tilde{\bm y}_1,\dots,\tilde{\bm y}_n$ be i.i.d. datapoints drawn from a distribution with mean $\bm\mu$ and covariance matrix $\Sigma \preccurlyeq \sigma^2 I$.
Let $\bm G$ be the set defined in \eqref{eq:Gdef}. Let $\bar{\bm x}$ be the average of datapoints in $\bm G$. Then the following holds on the event $\mathcal{E}_1\cap\mathcal{E}_3\cap\mathcal{E}_4$, where the events are defined in \eqref{event:1}, \eqref{event:3} and \eqref{event:4}:
\begin{equation}
    \lVert \bar{\bm x}-\bm\mu \rVert_2 \leq \sigma\sqrt{\alpha\delta}\para*{1+2\sqrt{\frac{c_1'}{e\log(d/\delta)}}}.
\end{equation}
\end{lemma}
\begin{proof}
Note that 

\begin{align}
    &\left\lVert \frac{|\bm G|}{n}(\bar{\bm x} - \bm\mu) \right\rVert_2\\
    =& \left\lVert \frac{1}{n}\sum\limits_{i=1}^n (\tilde{\bm y}_i-\bm\mu)- \frac{1}{n} \sum\limits_{i=1}^n (\tilde{\bm y}_i - \bm\mu)\mathbbm{1}\bre*{ \lVert \tilde{\bm y}_i-\bm\mu \rVert_2> \sigma\sqrt{\frac{d}{\alpha\delta}}} \right\rVert_2\\
    \leq  &\left\lVert \frac{1}{n}\sum\limits_{i=1}^n (\tilde{\bm y}_i-\bm\mu) \right\rVert_2+\left\lVert \frac{1}{n} \sum\limits_{i=1}^n \bm z_i \right\rVert_2\\
    \leq & \left\lVert \frac{1}{n}\sum\limits_{i=1}^n (\tilde{\bm y}_i-\bm\mu) \right\rVert_2+\left\lVert \frac{1}{n}\sum\limits_{i=1}^n (\bm z_i-E[\bm z_1]) \right\rVert_2 + \left\lVert E[\bm z_1] \right\rVert_2,\label{eq:start}
\end{align}
where $\bm z_i=(\tilde{\bm y}_i - \bm\mu)\mathbbm{1}\bre*{ \lVert \tilde{\bm y}_i-\bm\mu \rVert_2> \sigma\sqrt{\frac{d}{\alpha\delta}}}$. 

The last term is upper bounded as follows,
\begin{align}
    \left\lVert E[\bm z_1]\right\rVert_2=  &\left\lVert E\left[(\tilde{\bm y}_1 - \bm\mu)\mathbbm{1}\left\lbrace \lVert \tilde{\bm y}_1-\bm\mu \rVert_2> \sigma\sqrt{\frac{d}{\alpha\delta}} \right\rbrace \right]\right\rVert_2\\
    = & \max_{\lVert v\rVert_2=1}  v^\top E\left[(\tilde{\bm y}_1 - \bm\mu)\mathbbm{1}\left\lbrace \lVert \tilde{\bm y}_1-\bm\mu \rVert_2> \sigma\sqrt{\frac{d}{\alpha\delta}} \right\rbrace \right]\\
    = &\max_{\lVert v\rVert_2=1}   E\left[v^\top(\tilde{\bm y}_1 - \bm\mu)\mathbbm{1}\left\lbrace \lVert \tilde{\bm y}_1-\bm\mu \rVert_2> \sigma\sqrt{\frac{d}{\alpha\delta}} \right\rbrace \right]\\
    \mathop{\leq}\limits^{\text{(a)}} & \max_{\lVert v\rVert_2=1} \sqrt{E[v^\top (\tilde{\bm y}_1-\bm\mu)]^2P\para*{\lVert \tilde{\bm y}_1-\bm\mu \rVert_2> \sigma\sqrt{\frac{d}{\alpha\delta}} }}\label{eq:cs}\\
    = &\sqrt{\lambda_{\max}\left(\Sigma \right)P\para*{\lVert \tilde{\bm y}_1-\bm\mu \rVert_2> \sigma\sqrt{\frac{d}{\alpha\delta}} }}\\
    \mathop{\leq}\limits^{\text{(b)}} &\sqrt{\sigma^2\alpha\delta}\\
    = & \sigma\sqrt{\alpha\delta}.\label{eq:Ez}
\end{align}
The inequality (a) follows from Cauchy-Schwarz inequality, and (b) follows from Markov's inequality.

From  \eqref{eq:start}, \eqref{eq:prob1}, \eqref{eq:event34}, and \eqref{eq:Ez}, we get that on the event $\mathcal{E}_1\cap \mathcal{E}_3\cap \mathcal{E}_4$,
    \begin{equation}
    \lVert \bar{\bm x}-\bm\mu \rVert_2 \leq \sigma\sqrt{\alpha\delta}\para*{1+2\sqrt{\frac{c_1'}{e\log(d/\delta)}}}.
\end{equation}
\end{proof}

\begin{lemma}\label{lem:threshold}
Let $0< \tau\leq 1$. Suppose $\bm h\in \mathbb{R}^n$ such that $\forall i,\;0\leq h_i\leq 1$, and $\|\bm h\|_1\leq \epsilon n$ for some $\epsilon\in[0,1)$. Then
\begin{align}
    \sum\limits_{i=1}^n (1-h_i)1\{h_i\leq\tau\}\geq \para*{1-\frac{\epsilon}{\tau}}n.
\end{align} 
\begin{proof}
We first show that $\sum\limits_{i=1}^n 1\{h_i> \tau\}\leq \frac{\epsilon n}{\tau}$. Observe that
\begin{align}
    \epsilon n\geq \sum\limits_{i=1}^n h_i &= \sum\limits_{i=1}^n h_i1\{h_i\leq\tau\}+\sum\limits_{i=1}^n h_i1\{h_i> \tau\}\\
    &\geq \tau\sum\limits_{i=1}^n 1\{h_i> \tau\}.
\end{align}
Hence, we have
\begin{align}
    \sum\limits_{i=1}^n 1\{h_i> \tau\}\leq \frac{\epsilon n}{\tau}.
\end{align}
Consequently, we obtain
 \begin{align}
     \sum\limits_{i=1}^n (1-h_i)1\{h_i\leq\tau\} &=  \sum\limits_{i=1}^n (1-h_i)-\sum\limits_{i=1}^n (1-h_i)1\{h_i>\tau\}\\
     &\geq  \sum\limits_{i=1}^n (1-h_i)-(1-\tau)\sum\limits_{i=1}^n 1\{h_i>\tau\}\\
     &\geq (1-\epsilon)n -\frac{(1-\tau)\epsilon n}{\tau}\\
     &=\para*{1-\frac{\epsilon}{\tau}}n.
 \end{align}
\end{proof}
\end{lemma}

\subsection{Proof of Theorem \ref{thm:l0} }
\begin{proof}
Let $(\hat{\bm h},\hat{\bm x})$ be a feasible pair for \eqref{obj} lying in $\mathcal{S}$. Note that we get a corresponding feasible pair lying in $\mathcal{S}$ by only setting non-zero $\hat h_i$ to be 1. With slight abuse of notation, let $(\hat{\bm h},\hat{\bm x})$
be this feasible pair.

Let $\epsilon'\triangleq \alpha+\epsilon$. Let $\hat{\bm w}=\frac{\bm 1-\hat{\bm h}}{n-\|\hat{\bm h}\|_0}$ and $\beta=\|\hat{\bm h}\|_0/n$. Note that $\hat{\bm w}\in \Delta_{n,\beta}$. Consider $\bm h^*$ as defined in~\eqref{eq:outlier}. Let $\bar{\bm x}^*$ be the average of datapoints in the set $\bm G^*$ defined in \eqref{eq:G2} and let $\bm w^*=\frac{\bm 1 -\bm h^*}{n-\|\bm h^*\|_0}$. Observe that on event $\mathcal{E}_1$, $\bm w^*\in\Delta_{n,\epsilon'}$. From Lemma \ref{lem:tv}, we obtain
\begin{align}
    \mathrm{TV}(\hat{\bm w},\bm w^*)\leq \frac{\max(\beta,\epsilon')}{1-\min(\beta,\epsilon')}.
\end{align}
As a consequence of Lemma \ref{lem:wmin}, on event $\mathcal{E}_1\cap\mathcal{E}_2$, we have
\begin{align}
    \lambda_{\max}(\Sigma_{\bm w^*}) &\leq \lambda_{\mathrm{max}}\left(\frac{1}{\abs*{\bm I^*}}\sum_{i=1}^{n} (1-h_i^*)(\bm y_i-\bm\mu)(\bm y_i-\bm\mu)^\top\right)\leq \frac{c_1^2\sigma^2n}{|\bm I^*|}\leq \frac{c_1^2\sigma^2}{1-\epsilon'}\\
    \lambda_{\max}(\Sigma_{\hat{\bm w}}) &\leq \lambda_{\mathrm{max}}\left(\frac{1}{n-\|\hat{\bm h}\|_0}\sum_{i=1}^{n} (1-\hat{h}_i)(\bm y_i-\hat{\bm x})(\bm y_i-\hat{\bm x})^\top\right)\leq \frac{c_1^2\sigma^2n}{n-\|\hat{\bm h}\|_0}= \frac{c_1^2\sigma^2}{1-\beta}.
\end{align}

\noindent Consider the case $\beta\leq \epsilon'<1/2$. This implies $\mathrm{TV}(\hat{\bm w},\bm w^*)\leq \frac{\epsilon'}{1-\beta}<1$. From Lemma \ref{lem:res}, on event $\mathcal{E}_1\cap\mathcal{E}_2$, we get
\begin{align}
    \|\hat{\bm x}-\bar{\bm x}^*\|_2 &\leq \para*{\sqrt{\frac{c^2_1\sigma^2}{1-\epsilon'}}+\sqrt{\frac{c^2_1\sigma^2}{1-\beta}}}\sqrt{\frac{\epsilon'}{1-\epsilon'-\beta}}\\
    &\leq \frac{2c_1\sigma\sqrt{\epsilon'}}{\sqrt{(1-\epsilon')(1-2\epsilon')}}.
\end{align}

Consider the case $\epsilon'\leq \beta< 1-\epsilon'$. This implies $\mathrm{TV}(\hat{\bm w},\bm w^*)\leq \frac{\beta}{1-\epsilon'}<1$. From Lemma \ref{lem:res}, on event $\mathcal{E}_1\cap\mathcal{E}_2$, we get
\begin{align}
    \|\hat{\bm x}-\bar{\bm x}^*\|_2 &\leq \para*{\sqrt{\frac{c^2_1\sigma^2}{1-\epsilon'}}+\sqrt{\frac{c^2_1\sigma^2}{1-\beta}}}\sqrt{\frac{\beta}{1-\epsilon'-\beta}}.
\end{align}

Consequently, on the event $\mathcal{E}$ defined in \eqref{event:all}, using Lemma \ref{lem:conc}, \eqref{eq:res} and applying triangle inequality, we obtain that with probability at least $1-4\delta$
\begin{equation*}
    \| \hat{\bm x}-\bm\mu\|_2 \leq \|\hat{\bm x}-\bar{\bm x}^*\|_2+\sqrt{\frac{c_1^2\sigma^2}{1-\alpha}.\frac{\epsilon}{1-\epsilon}}+\sigma\sqrt{\alpha\delta}\para*{1+2\sqrt{\frac{c_1'}{e\log(d/\delta)}}}.
\end{equation*}
It follows from \eqref{event2} that on the event $\mathcal{E}_2$, $\mathcal{E}_2$, $(\bm h^{*}, \bar{\bm x}^*)$ is feasible. We also have that $\mathcal{E}_1=\{|\bm I| \geq (1-\alpha) n\}\subseteq\{n-\|\bm h^*\|_0\geq (1-\epsilon')n\} =\{\|\bm h^*\|_0 \leq \epsilon' n\}$. Note that for any globally optimal solution of \eqref{obj}, by setting all its non-zero $h_i$ to be 1, we can always get corresponding feasible and globally optimal $(\bm h^{\mathrm{opt}}, \bm x^{\mathrm{opt}})$ with $h_i^{\mathrm{opt}} \in \{0,1\}$ and ${\bm x^{\mathrm{opt}}}=\frac{\sum_{\{i:{h}_i^{\mathrm{opt}}=0\}} \bm y_i}{|\{i:{h}_i^{\mathrm{opt}}=0\}|}$ (i.e., $\bm x^{\mathrm{opt}}$ is the average of the $\bm y_i$'s corresponding to $h_i^{\mathrm{opt}}=0$), and the objective value remains unchanged. Since $(\bm h^{\mathrm{opt}}, \bm x^{\mathrm{opt}})$ is globally optimal, and $(\bm h^{*}, \bm\mu)$ is feasible, we have $\|\bm h^{\mathrm{opt}}\|_0 \leq \|\bm h^*\|_0 \leq \epsilon' n$. Hence, $(\bm h^{\mathrm{opt}}, \bm x^{\mathrm{opt}})\in\mathcal{S}'$ with $\|\bm h^{\mathrm{opt}}\|_0 \leq \epsilon' n$.

\end{proof}

\subsection{Proof of Theorem \ref{thm:lp2}}
\begin{proof}
Let $(\hat{\bm h},\hat{\bm x})\in \mathcal{S}'$ be a feasible pair for \eqref{objL1} with some $0<p\leq 1$. We have
\begin{equation}
    \|\hat{\bm h}\|_p\leq ((1-\epsilon')n)^{1/p}.
\end{equation}
Since $0\leq \hat h_i\leq 1$ for all $i$, we have
\begin{align}
    &\bra*{\sum\limits_{i=1}^n \hat h_i}^{1/p}\leq \bra*{\sum\limits_{i=1}^n \hat h_i^p}^{1/p}\leq ((1-\epsilon')n)^{1/p}.
\end{align}
This implies the following
\begin{align}
    \|\hat{\bm h}\|_1 &\leq \|\hat{\bm h}\|^p_p\leq  (1-\epsilon')n\\
    \|\bm 1-\hat{\bm h}\|_1 &\geq n- \|\hat{\bm h}\|^p_p\geq  \epsilon'n.
\end{align}
Let $\hat{\bm w}=\frac{\bm 1-\hat{\bm h}}{\|\bm 1- \hat{\bm h}\|_1}$ and $\beta=\|\hat{\bm h}\|^p_p/n$.  Note that $\hat{\bm w}\in \Delta_{n,\beta}$. Consider $\bm h^*$ as defined in~\eqref{eq:outlier}. Let $\bar{\bm x}^*$ be the average of datapoints in the set $\bm G^*$ defined in \eqref{eq:G2} and let $\bm w^*=\frac{\bm 1 -\bm h^*}{n-\|\bm h^*\|_0}$. Observe that on event $\mathcal{E}_1$, $\bm w^*\in\Delta_{n,\epsilon'}$.
As a consequence of Lemma \ref{lem:wmin}, on event $\mathcal{E}_1\cap\mathcal{E}_2$, we have
\begin{align}
    \lambda_{\max}(\Sigma_{\bm w^*}) &\leq \lambda_{\mathrm{max}}\left(\frac{1}{\abs*{\bm I^*}}\sum_{i=1}^{n} (1-h_i^*)(\bm y_i-\bm\mu)(\bm y_i-\bm\mu)^\top\right)\leq \frac{c_1^2\sigma^2n}{|\bm I^*|}\leq \frac{c_1^2\sigma^2}{1-\epsilon'}\\
    \lambda_{\max}(\Sigma_{\hat{\bm w}})  &\leq \lambda_{\mathrm{max}}\left(\frac{1}{n-\|\bm h\|_1}\sum_{i=1}^{n} (1-\hat{h}_i)(\bm y_i-\hat{\bm x})(\bm y_i-\hat{\bm x})^\top\right)\leq \frac{c_1^2\sigma^2n}{n-\|\bm h\|^p_p}= \frac{c_1^2\sigma^2}{1-\beta}.
\end{align}
From Lemma \ref{lem:tv}, we obtain
\begin{align}
    \mathrm{TV}(\hat{\bm w},\bm w^*)\leq \frac{\max(\beta,\epsilon')}{1-\min(\beta,\epsilon')}.
\end{align}

Consider the case $\beta\leq \epsilon'<1/2$. This implies $\mathrm{TV}(\hat{\bm w},\bm w^*)\leq \frac{\epsilon'}{1-\beta}<1$. From Lemma \ref{lem:res}, on event $\mathcal{E}_1\cap\mathcal{E}_2$, we get
\begin{align}
    \|\hat{\bm x}-\bar{\bm x}^*\|_2 &\leq \para*{\sqrt{\frac{c^2_1\sigma^2}{1-\epsilon'}}+\sqrt{\frac{c^2_1\sigma^2}{1-\beta}}}\sqrt{\frac{\epsilon'}{1-\epsilon'-\beta}}\\
    &\leq \frac{2c_1\sigma\sqrt{\epsilon'}}{\sqrt{(1-\epsilon')(1-2\epsilon')}}.
\end{align}

Consider the case $\epsilon'\leq \beta< 1-\epsilon'$. This implies $\mathrm{TV}(\hat{\bm w},\bm w^*)\leq \frac{\beta}{1-\epsilon'}<1$. From Lemma \ref{lem:res}, on event $\mathcal{E}_1\cap\mathcal{E}_2$, we get
\begin{align}
    \|\hat{\bm x}-\bar{\bm x}^*\|_2 &\leq \para*{\sqrt{\frac{c^2_1\sigma^2}{1-\epsilon'}}+\sqrt{\frac{c^2_1\sigma^2}{1-\beta}}}\sqrt{\frac{\beta}{1-\epsilon'-\beta}}.
\end{align}

Consequently, on the event $\mathcal{E}$ defined in \eqref{event:all}, using Lemma \ref{lem:conc}, \eqref{eq:res} and applying triangle inequality, we obtain that with probability at least $1-4\delta$
\begin{equation*}
    \| \hat{\bm x}-\bm\mu\|_2 \leq \|\hat{\bm x}-\bar{\bm x}^*\|_2+\sqrt{\frac{c_1^2\sigma^2}{1-\alpha}.\frac{\epsilon}{1-\epsilon}}+\sigma\sqrt{\alpha\delta}\para*{1+2\sqrt{\frac{c_1'}{e\log(d/\delta)}}}.
\end{equation*}
Let $({\bm h}^{\mathrm{opt}},{\bm x}^{\mathrm{opt}})$ be an optimal solution to \eqref{objL1}. From Lemma \ref{lem:wmin} we have that $\left({\bm h}^{\mathrm{opt}},\frac{\sum\limits_{i=1}^n (1-h_i^{\mathrm{opt}}) \bm y_i}{\sum\limits_{i=1}^n (1-h_i^{\mathrm{opt}})}\right)$ is also an optimal solution. Note that on the event $\mathcal{E}$, we have that $(\bm h^*,\bm\mu)$ is a feasible pair for \eqref{objL1}. Hence,
\begin{equation}
    \|{\bm h}^{\mathrm{opt}}\|_p\leq \|\bm h^*\|_p\leq (\epsilon'n)^{1/p}.
\end{equation}
    This implies
    \begin{align}
    \left({\bm h}^{\mathrm{opt}},\frac{\sum\limits_{i=1}^n (1-h_i^{\mathrm{opt}}) \bm y_i}{\sum\limits_{i=1}^n (1-h_i^{\mathrm{opt}})}\right)\in\mathcal{S}'.
\end{align}
\end{proof}

\subsection{Proof of Theorem \ref{thm:algo}}
\begin{proof} 
We prove the result by the method of induction. 

Let $\bm x^{(0)}$ be the coordinate-wise median of the corrupted sample. It is easy to check that under the conditions stated in Theorem \ref{thm:algo}, it follows that $c_3\leq c_1$ and $\epsilon\leq (1-\alpha)(1-\epsilon)$. Note that if $n\geq 90\log\para*{\frac{d}{\delta}}$, then by Lemma \ref{lem:median}, Lemma \ref{lem:conc}, \eqref{eq:res} and triangle inequality, we have that the following holds with probability at least $1-\delta$:
\begin{align}
\|\bm x^{(0)}-\bar{\bm x}^*\|_2&= \|\bm x^{(0)}-\bm \mu +\bm \mu - \bar{\bm x}+ \bar{\bm x}-\bar{\bm x}^*\|_2 \\
&\leq \|\bm x^{(0)}-\bm \mu\|_2+\|\bm \mu - \bar{\bm x}\|_2+\|\bar{\bm x}-\bar{\bm x}^*\|_2 \\
&\leq 3\sigma\sqrt{d}+\sigma c_3+\sigma\sqrt{\frac{c_1^2\sigma^2}{1-\alpha}.\frac{\epsilon}{1-\epsilon}}\leq \sigma(3\sqrt{d}+2c_1)=\sigma c_2^{(0)}.
\end{align}
Let $\mathcal{E}'$ be the event
\begin{align}
    \|\bm x^{(0)}-\bar{\bm x}^*\|_2\leq \sigma c_2^{(0)}.
\end{align}

All the following statements hold on the event $\mathcal{E}\cup \mathcal{E}'$, where $\mathcal{E}$ is defined in \eqref{event:all}. Also note that $\pr(\mathcal{E}\cup\mathcal{E}')\geq 1-5\delta$, when {$n\geq \max\left\lbrace 90,\frac{2e}{c'_1\delta^2}d\right\rbrace\log\para*{\frac{d}{\delta}}$}.

Suppose $\|\bm x^{(t)}-\bar{\bm x}^*\|_2 \leq c^{(t)}_2\sigma$ and $\|\bm h^{(t-1)}\|^p_p\leq\epsilon'n$. Let $\bm h^{(t)}$ be an optimal solution to
\begin{align}
    &\min_{\bm h} \|\bm h\|_p\\
    \text{s.t. } & \lambda_{\max} \left(\sum_{i=1}^n (1-h_{i}) (\bm y_i -\bm x^{(t)})(\bm y_i - \bm x^{(t)})^\top\right)\leq \left(c_1^2+(c^{(t)}_2)^2\right)\sigma^2n\\
    &0\leq h_i\leq 1,\;\forall i\in [n].
\end{align}
From Lemma \ref{lem:feasible}, we have that $\bm h^*$ is a feasible point for the above optimization problem. Hence,
\begin{equation}
    \|\bm h^{(t)}\|_p\leq \|\bm h^*\|_p\leq (\epsilon'n)^{1/p}.
\end{equation}
Since $0\leq h^{(t)}_i\leq 1$ for all $i$, we have
\begin{align}
    &\bra*{\sum\limits_{i=1}^n  h_i^{(t)}}^{1/p}\leq \bra*{\sum\limits_{i=1}^n \left( h^{(t)}_i\right)^p}^{1/p}\leq (\epsilon'n)^{1/p}.
    \end{align}
    This implies
    \begin{align}
     \|\bm h^{(t)}\|_1 \leq \epsilon'n.
\end{align}
Let $\bm w$ be such that 
\begin{equation}
w_i= \frac{(1-h^{(t)}_i) 1\{ h^{(t)}_i\leq\tau\}}{\sum\limits_{i=1}^n (1- h^{(t)}_i) 1\{h^{(t)}_i\leq\tau\}}.    
\end{equation}
By Lemma \ref{lem:threshold}, we have that $\bm w\in \Delta_{n,\frac{\epsilon'}{\tau}}$. Now we follow the proof of Theorem \ref{thm:lp2}. 
Let $\bm x^{(t+1)}=\sum\limits_{i=1}^n w_i \bm y_i$. Observe that $\bm w^*\in\Delta_{n,\epsilon'}$. 
As a consequence of Lemma \ref{lem:wmin}, we have
\begin{align}
    \lambda_{\max}(\Sigma_{\bm w^*})&\leq \lambda_{\mathrm{max}}\left(\frac{1}{\abs*{\bm I^*}}\sum_{i=1}^{n} (1-h_i^*)(\bm y_i-\bm\mu)(\bm y_i-\bm\mu)^\top\right)\leq \frac{c_1^2\sigma^2n}{|\bm I^*|}\leq \frac{c_1^2\sigma^2}{1-\epsilon'},\\
    \lambda_{\max}(\Sigma_{\bm w})&= \lambda_{\mathrm{max}}\left(\frac{1}{\sum\limits_{i=1}^n (1-h^{(t)}_i) 1\{h^{(t)}_i\leq\tau\}}\sum_{i=1}^{n} (1- h^{(t)}_i)1\{h^{(t)}_i\leq\tau\}(\bm y_i-\bm x^{(t+1)})(\bm y_i-\bm x^{(t+1)})^\top\right)\\
    &\leq\lambda_{\mathrm{max}}\left(\frac{1}{\sum\limits_{i=1}^n (1-h^{(t)}_i) 1\{h^{(t)}_i\leq\tau\}}\sum_{i=1}^{n} (1-h^{(t)}_i)1\{h^{(t)}_i\leq\tau\}(\bm y_i-\bm x^{(t)})(\bm y_i-\bm x^{(t)})^\top\right)\\
    &\leq \lambda_{\mathrm{max}}\left(\frac{1}{(1-\epsilon'/\tau)n}\sum_{i=1}^{n} (1-h^{(t)}_i)(\bm y_i-\bm x^{(t)})(\bm y_i-\bm x^{(t)})^\top\right)\\
    &\leq\frac{\left(c_1^2+(c^{(t)}_2)^2\right)\sigma^2}{1-\epsilon'/\tau}.
\end{align}
From Lemma \ref{lem:tv}, we obtain
\begin{align}
     \mathrm{TV}(\bm w,\bm w^*)\leq \frac{\max(\epsilon'/\tau,\epsilon')}{1-\min(\epsilon'/\tau,\epsilon')}=\frac{\epsilon'/\tau}{1-\epsilon'}.
\end{align}
From Lemma \ref{lem:res}, we get
\begin{align}
    \|\bm x^{(t+1)}-\bar{\bm x}^*\|&\leq (\sqrt{\lambda_{\max}(\Sigma_{\bm w})}+\sqrt{\lambda_{\max}(\Sigma_{\bm w^*})})\sqrt{\frac{\mathrm{TV}(\bm w,\bm w^*)}{1-\mathrm{TV}(\bm w,\bm w^*)}}\\
    &\leq \para*{\sqrt{\frac{\left(c_1^2+(c^{(t)}_2)^2\right)\sigma^2}{1-\frac{\epsilon'}{\tau}}}+\sqrt{\frac{c^2_1\sigma^2}{1-\epsilon'}}}\sqrt{\frac{\epsilon'/\tau}{1-\epsilon'-\epsilon'/\tau}}\\
    &\leq \para*{\frac{(c_1+c^{(t)}_2)\sigma}{\sqrt{1-\epsilon'/\tau}}+\frac{c_1\sigma}{\sqrt{1-\epsilon'}}}\sqrt{\frac{\epsilon'/\tau}{1-\epsilon'-\epsilon'/\tau}}\\
    &=\sigma(\gamma(\epsilon') c^{(t)}_2+\beta(\epsilon'))\\
    &\leq \sigma(\gamma(\check{\epsilon})c_2^{(t)}+\beta(\check{\epsilon}))\\
    &=\sigma c_2^{(t+1)}.
\end{align}

We established that $\|\bm x^{(t+1)}-\bar{\bm x}^*\|_2\leq \sigma c_2^{(t+1)}$ and $\|\bm h^{(t)}\|_p^p\leq \epsilon'n$. Hence, by the principle of mathematical induction, the result follows. It is easy to check that $\gamma(\check{\epsilon})<1$ holds if and only if $\check{\epsilon}<f(\tau)$. Furthermore, $\check{\epsilon}<f(\tau)$ implies $\check{\epsilon}<\tau$. Thus, we have that
\begin{align}
    \|\bm x^{(t)}-\bar{\bm x}^*\|_2 &\leq \sigma \left(\gamma(\epsilon') c_2^{(t-1)}+\beta(\epsilon') \right)\\
    &= \sigma\left[\gamma(\epsilon')\para*{c_2^{(0)}\gamma(\check{\epsilon})^{t-1}+\frac{1-\gamma(\check{\epsilon})^{t-1}}{1-\gamma(\check{\epsilon})}\beta(\check{\epsilon})}+\beta(\epsilon')\right].
  \end{align}  

Consequently, using Lemma \ref{lem:conc}, \eqref{eq:res} and applying triangle inequality, we obtain that with probability at least $1-5\delta$
\begin{equation*}
    \| \bm x^{(t)}-\bm\mu\|_2 \leq \sigma\left[\gamma(\epsilon')\para*{c_2^{(0)}\gamma(\check{\epsilon})^{t-1}+\frac{1-\gamma(\check{\epsilon})^{t-1}}{1-\gamma(\check{\epsilon})}\beta(\check{\epsilon})}+\beta(\epsilon')\right]+c_1\sigma\sqrt{\frac{\epsilon}{(1-\alpha)(1-\epsilon)}}+c_3\sigma.
\end{equation*}

It is easy to see that for $T=1+\frac{\log c_2^{(0)}}{|\log \gamma(\check{\epsilon})|}$, we have
\begin{align}
    \|\bm x^{(T)}-\bm\mu\|_2\leq \sigma\left[\gamma(\epsilon')\para*{1+\frac{\beta(\check{\epsilon})}{1-\gamma(\check{\epsilon})}}+\beta(\epsilon')\right]+c_1\sigma\sqrt{\frac{\epsilon}{(1-\alpha)(1-\epsilon)}}+c_3\sigma= O(\sigma\sqrt{\epsilon'}).
\end{align}

\end{proof}

\subsection{Solving \texorpdfstring{$\ell_1$}{l1} objective via Packing SDP}
\begin{align}\label{obj1}
   \min_{ \bm h} & \|\bm h\|_1 \\
    s.t. \ & 0\leq h_i \leq 1, \forall i, \nonumber\\
   & \lambda_{\mathrm{max}} \left(\sum_{i=1}^{n} (1-h_i)(\bm y_i-\bm x)(\bm y_i-\bm x)^\top\right)\leq c n \sigma^2. \nonumber
\end{align}

Define the vector $\bm w$ with $w_i\triangleq 1-h_i$. Since $0\leq h_i \leq 1$, we have $0\leq w_i \leq 1$. Further, $\| \bm h \|_1= \sum_{i=1}^{n} h_i=\sum_{i=1}^{n} (1-w_i)=n- \sum_{i=1}^{n} w_i =n- \bm 1^\top \bm w$. Therefore, solving \eqref{obj1} is equivalent to solving the following:
\begin{align}\label{obj1w}
   \max_{\bm w} & \ \bm 1^\top \bm w \\
    s.t. \ & 0\leq w_i \leq 1, \forall i, \nonumber\\
   & \lambda_{\mathrm{max}}\left(\sum_{i=1}^{n} w_i(\bm y_i-\bm x)(\bm y_i-\bm x)^\top\right)\leq c n \sigma^2. \nonumber
\end{align}

Then, we rewrite the constraints $0 \leq w_i \leq 1, \forall i$ as $0 \leq w_i$, and $\sum w_i e_ie_i^\top \preceq I_{n\times n}$, where $e_i$ is the $i$-th standard basis vector in $\mathbb{R}^n$. 
This establishes the equivalence between \eqref{obj1w} and \eqref{packing SDP}.

\subsection{Minimizing \texorpdfstring{$\ell_p$}{lp} via iterative re-weighted \texorpdfstring{$\ell_2$}{l2}}
Consider $\ell_p$ ($0<p<1$) in Step 1 of Algorithm \ref{alg:1}. We have the following equivalent objective:
\begin{align}\label{objp}
   \min_{ \bm h} & \|\bm h\|_p^p \\
    s.t. \ & 0\leq h_i \leq 1, \forall i, \nonumber\\
   & \lambda_{\mathrm{max}}\left(\sum_{i=1}^{n} (1-h_i)(\bm y_i-\bm x)(\bm y_i-\bm x)^\top\right)\leq  c\sigma^2n. \nonumber
\end{align}
Note that $\|\bm h\|_p^p=\sum_{i=1}^{n} h_i^p=\sum_{i=1}^{n} (h_i^2)^{\frac{p}{2}}$. Consider that we employ the iterative re-weighted $\ell_2$ technique~\cite{4518498,gorodnitsky1997sparse}. Then at $(k+1)$-th inner iteration, we construct a tight upper bound on $\|\bm h\|_p^p$ at ${\bm h^{(k)}}^2$ as
\begin{equation}
    \sum_{i=1}^{n} \left\lbrack{\left({h_i^{(k)}}^2\right)}^{\frac{p}{2}}+\frac{p}{2}{\left({h_i^{(k)}}^2\right)}^{\frac{p}{2}-1}\left(h_i^2-{h_i^{(k)}}^2\right)\right\rbrack.
\end{equation} 
We minimize this upper bound:
\begin{align}\label{objp_t}
   \min_{ \bm h} & \sum_{i=1}^{n}{\left({h_i^{(k)}}^2\right)}^{\frac{p}{2}-1}h_i^2 \\
    s.t. \ & 0\leq h_i \leq 1, \forall i, \nonumber\\
   & \lambda_{\mathrm{max}}\left(\sum_{i=1}^{n} (1-h_i)(\bm y_i-\bm x)(\bm y_i-\bm x)^\top\right)\leq c n \sigma^2, \nonumber
\end{align}
Define $u_i={\left({h_i^{(k)}}\right)}^{\frac{p}{2}-1}$, the objective in~\eqref{objp_t} becomes $\sum_{i=1}^{n} u_i^2h_i^2$. Define the vector $\bm w$ with $w_i\triangleq 1-h_i$. Since $0\leq h_i \leq 1$, we have $0\leq w_i \leq 1$. Further, $\sum_{i=1}^{n} u_i^2 h_i^2= \sum_{i=1}^{n} u_i^2 (1-w_i)^2=\sum_{i=1}^{n} (u_i-u_iw_i)^2$. So, solving~\eqref{objp_t} is equivalent to solving the following:
\begin{align}\label{objp_t_w}
   \min_{ \bm w} & \sum_{i=1}^{n}(u_i-u_iw_i)^2 \\
    s.t. \ & 0\leq w_i \leq 1, \forall i, \nonumber\\
   & \lambda_{\mathrm{max}}(\sum_{i=1}^{n} w_i(\bm y_i-\bm x)(\bm y_i-\bm x)^\top)\leq c n \sigma^2. \nonumber
\end{align}

Further, define the vector $\bm z$ with $z_i\triangleq u_iw_i$. Then solving~\eqref{objp_t_w} is equivalent to solving the following:
\begin{align}\label{objp_t_z}
   \min_{ \bm z} & \| \bm u-\bm z\|_2^2 \\
    s.t. \ & 0\leq z_i \leq u_i, \forall i, \nonumber\\
   & \lambda_{\mathrm{max}}\left (\sum_{i=1}^{n} z_i[(\bm y_i-\bm x)(\bm y_i-\bm x)^\top/u_i]\right)\leq c n \sigma^2. \nonumber
\end{align}

Then, we rewrite the constraints $0 \leq z_i \leq u_i, \forall i$ as $0 \leq z_i$, and $\sum_{i=1}^{n} z_i e_ie_i^\top \preceq \diag (\bm u)$, where $e_i$ is the $i$-th standard basis vector in $\mathbb{R}^n$. Finally, we can turn~\eqref{objp_t_z} into the following least squares problem with semidefinite cone constraints:
\begin{align}\label{objp_t_z_SDP}
   \min_{ \bm z} & \| \bm u-\bm z\|_2^2 \\
    s.t. \ & \ z_i \geq 0,  \forall i, \nonumber\\
   & \sum_{i=1}^{n} z_i \begin{bmatrix}
    e_ie_i^\top &  \\
    &  (\bm y_i-\bm x)(\bm y_i-\bm x)^\top /u_i \end{bmatrix} \preceq \begin{bmatrix}
    \diag (\bm u)&  \\
    &  c n \sigma^2 I_{d\times d} \end{bmatrix} .\nonumber
\end{align}




\subsection{Solving weighted \texorpdfstring{$\ell_1$}{l1} objective via Packing SDP}
Consider $\ell_p$ $(0<p<1)$ in Step 1 of Algorithm \ref{alg:1} (see objective~\eqref{objp}). If we employ iterative  re-weighted $\ell_1$ approach~\cite{candes2008enhancing,4518498}, we need to solve the following problem:
\begin{align}\label{objp_wL1}
   \min_{ \bm h} & \sum_{i=1}^{n} u_i h_i\\
    s.t. \ & 0\leq h_i \leq 1, \forall i, \nonumber\\
   & \lambda_{\mathrm{max}}\left(\sum_{i=1}^{n} (1-h_i)(\bm y_i-\bm x)(\bm y_i-\bm x)^\top\right)\leq c n \sigma^2, \nonumber
\end{align}

where $u_i$ is the weight on corresponding $h_i$. Define the vector $\bm w$ with $w_i\triangleq 1-h_i$. Since $0\leq h_i \leq 1$, we have $0\leq w_i \leq 1$. Further, $\sum_{i=1}^{n} u_i h_i= \sum_{i=1}^{n} u_i (1-w_i)=\sum_{i=1}^{n} u_i- \sum_{i=1}^{n} u_i w_i$. So, solving \eqref{objp_wL1} is equivalent to solving the following:
\begin{align}\label{obj1wp}
   \max_{\bm w} & \ \bm u^\top \bm w \\
    s.t. \ & 0\leq w_i \leq 1, \forall i, \nonumber\\
   & \lambda_{\mathrm{max}}\left(\sum_{i=1}^{n} w_i(\bm y_i-\bm x)(\bm y_i-\bm x)^\top\right)\leq c n \sigma^2. \nonumber
\end{align}

Then, we rewrite the constraints $0 \leq w_i \leq 1, \forall i$ as $0 \leq w_i$, and $\sum w_i e_ie_i^\top \preceq I_{n\times n}$, where $e_i$ is the $i$-th standard basis vector in $\mathbb{R}^n$. Finally, we can turn \eqref{obj1wp} into the following Packing SDP:
\begin{align}\label{packing SDPw}
   \max_{\bm w} & \ \bm u^\top \bm w \\
    s.t. \ & w_i \geq 0, \forall i, \nonumber\\
   & \sum_{i=1}^{n} w_i \begin{bmatrix}
    e_ie_i^\top &  \\
    &  (\bm y_i-\bm x)(\bm y_i-\bm x)^\top \end{bmatrix} \preceq \begin{bmatrix}
    I_{n\times n} &  \\
    &  c n \sigma^2 I_{d\times d} \end{bmatrix} .\nonumber
\end{align}

\subsection{Corrupted image dataset}
We use real face images to test the effectiveness of the robust mean estimation methods. The average face of particular regions or certain groups of people is useful for many social and psychological studies~\cite{little2011facial}. Here we use 100 frontal human face images from Brazilian face database\footnote{https://fei.edu.br/~cet/facedatabase.html} as inliers. For the outliers, we choose 15 face images of cat and dog from CIFAR10~\cite{krizhevsky2009learning}. In order to run the CDG method~\cite{convex}, we scale the size of images to 18 $\times$ 15 pixels, so the dimension of each datapoint is 270. Fig.~\ref{fig:sample_inlier} and Fig.~\ref{fig:sample_outlier} show the sample inlier and outlier images. Fig.~\ref{fig:face_mean} shows the oracle solution (the average of the 100 inlier human faces) and the estimated mean by each method, as well as their $\ell_2$ distances to the oracle solution. The proposed $\ell_1$ and $\ell_p$ methods achieve smaller recovery error than the state-of-the-art methods. The estimated mean faces by the proposed methods also look visually similar to the oracle solution, which illustrates the efficacy of the proposed $\ell_1$ and $\ell_p$ methods. 

\begin{figure*}[!ht]
	\centering
	\includegraphics[trim=50 22 155 20, clip=true,width=0.3\linewidth]{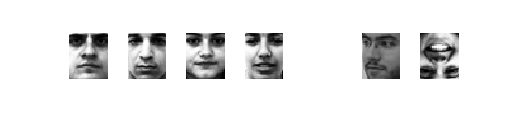}
	\captionsetup{justification=centering}
	\caption{Sample inlier human face images.}
	\label{fig:sample_inlier}
\end{figure*}

\begin{figure*}[!ht]
	\centering
	\includegraphics[trim=0 1200 10 300, clip=true,width=0.22\linewidth]{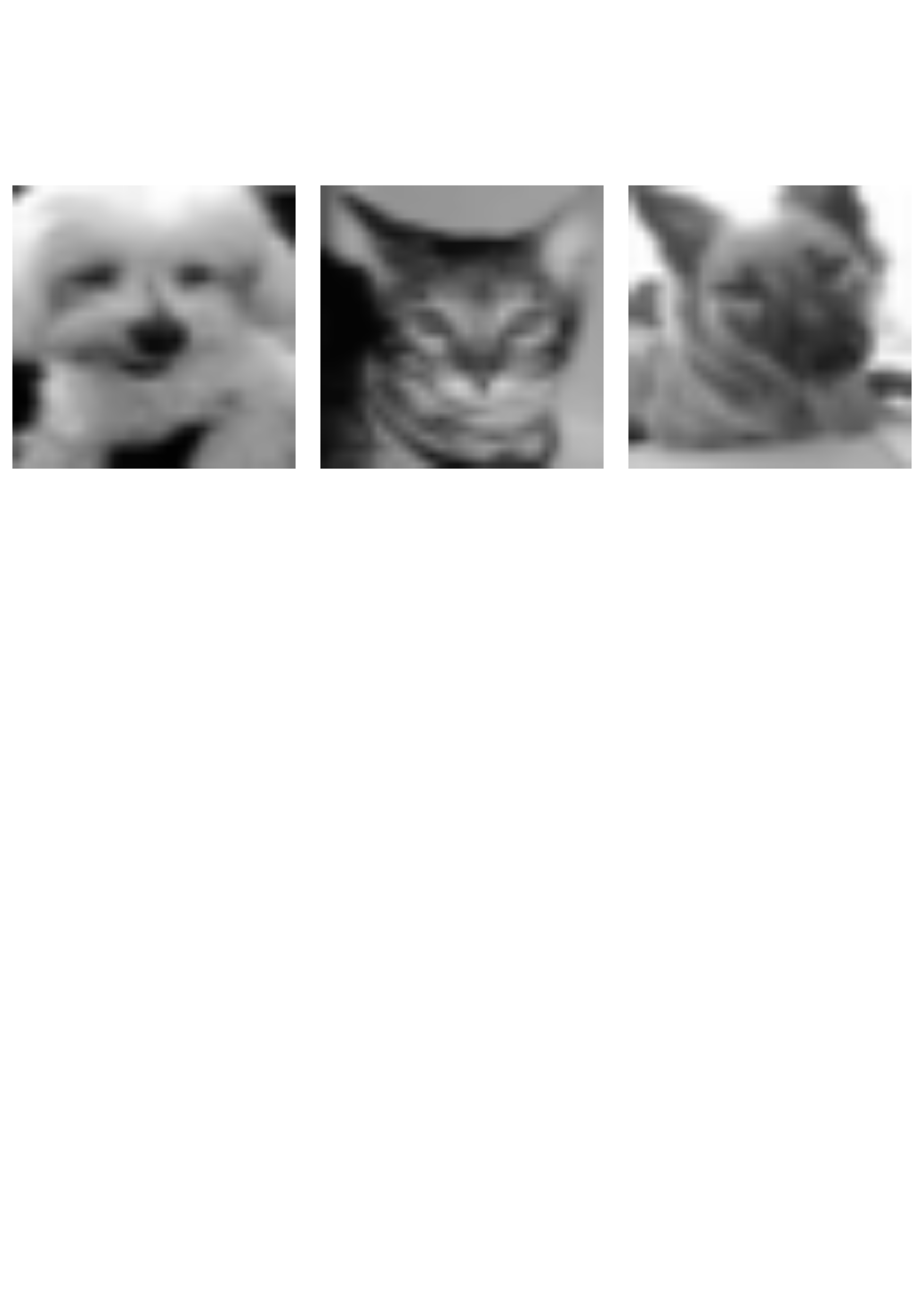}
	\captionsetup{justification=centering}
	\caption{Sample outlier cat and dog face images from CIFAR10.}
	\label{fig:sample_outlier}
\end{figure*}

\begin{figure*}[h]
	\centering
	\includegraphics[trim=40 146 24 110, clip=true,width=0.5\linewidth]{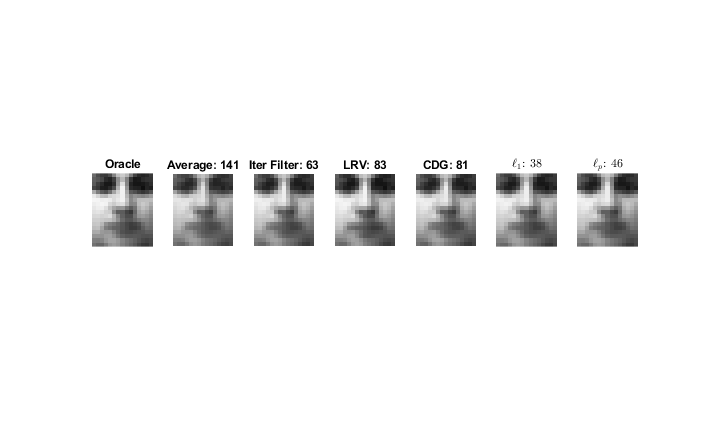}
	\captionsetup{justification=centering}
	\caption{Reconstructed mean face and its recovery error by each method.}
	\label{fig:face_mean}
\end{figure*}

\end{document}